\documentclass[11pt]{article}

\usepackage[dvipsnames,table]{xcolor}
\usepackage[english]{babel}
\usepackage{subcaption}
\usepackage{url}
\usepackage{amssymb}
\usepackage{stmaryrd}
\usepackage{amsmath}
\usepackage{mathrsfs}
\usepackage[thmmarks,amsmath]{ntheorem}
\usepackage[left]{lineno}
\usepackage{xfrac}
\usepackage{nicefrac}
\usepackage[textsize=scriptsize,backgroundcolor=yellow!40]{todonotes}
\usepackage[nonumberlist,acronym,sanitize=none]{glossaries}
\glsdisablehyper
\usepackage{comment}
\usepackage[colorlinks=true,
            linkcolor=purple, citecolor=teal, urlcolor=blue,
            bookmarks=true, bookmarksopen=true, bookmarksdepth=3,
            hyperfootnotes=true, hyperindex=true, pdfstartview=FitH,
            plainpages=false, pdfpagelabels=true]{hyperref}
\usepackage[capitalise,nameinlink]{cleveref}
\crefname{algocf}{Alg.}{Algs.}
\Crefname{algocf}{Algorithm}{Algorithms}
\crefname{section}{Sect.}{Sects.}
\Crefname{section}{Section}{Sections}
\crefname{definition}{Def.}{Defs.}
\Crefname{definition}{Definition}{Definitions}
\crefname{table}{Tab.}{Tabs.}
\Crefname{table}{Table}{Tables}

\usepackage{tkz-base}
\usetikzlibrary{decorations.pathmorphing,trees,snakes,arrows,shapes,automata,petri,fit}
\usepackage{paralist}
\usepackage{multicol}
\usepackage{booktabs}
\usepackage[inline]{enumitem}
\usepackage{multirow}
\usepackage{rotating}
\usepackage{etaremune}
\usepackage[ruled,linesnumbered,algo2e]{algorithm2e}

\SetCommentSty{mycommentfont}
\usepackage{marginnote}
\usepackage{mathtools}
\usepackage{adjustbox}
\usepackage{ifthen}
\usepackage[normalem]{ulem}
\usepackage{soul}
\usepackage{floatrow}
\usepackage{listings}
\crefname{lstlisting}{Listing}{Listings}
\usepackage{lipsum}
\usepackage{makecell}
\usepackage{diagbox}
\usepackage[scientific-notation=false,group-separator={,}]{siunitx}
\usepackage{microtype}
\usepackage[bottom]{footmisc}
\usepackage{leftindex}
\usepackage{xspace}
\usepackage{xparse}
\usepackage{orcidlink}

\usepackage{tabularx}
\usepackage{graphicx}
\usepackage{wrapfig}

\usepackage{float}

\usepackage{makeidx}         % allows index generation

\newcommand{\DclrSty}[1] {\textsc{#1}}
\def\Declare {\DclrSty{Declare}}
\def\DecSpec {\ensuremath{\mathcal{D}\hspace{-0.5ex}\mathcal{S}}}
\def\MinDecSpec {\ensuremath{\mathcal{S}_{\textrm{ined}}}}
\def\Subsum {\ensuremath{\sqsubseteq}}

\newcommand{\ActvOf}[1] {\ensuremath{{\begin{tikzpicture}[x=0.2ex,y=0.2ex,baseline={(current bounding box.center)}]\draw (0,0) rectangle (3,3);\end{tikzpicture}}#1}}
\newcommand{\TrgtOf}[1] {\ensuremath{#1{\begin{tikzpicture}[x=0.2ex,y=0.2ex,baseline={(current bounding box.center)}]\draw[-{Stealth[scale=0.75]}] (0,0) -- (1,0);\end{tikzpicture}}}}
\def\Cns {\ensuremath{\DclrSty{k}}}
\newacronym{rcns}{\textsc{rCon}}{Reactive Constraint}
\def\CnsSet {\ensuremath{K}}
\def\DecTemp {\ensuremath{\DclrSty{k}}}
\def\DecRepertoire {\ensuremath{\DclrSty{Rep}}}
\def\CnsInstRelation {\ensuremath{\Gamma}}

\newcommand{\CnsInterpFun} {\ensuremath{\mathscr{I}}}

\def\CnsEvalFunctor {\ensuremath{\eta}}
\def\Supp {\ensuremath{\mathrm{supp}}}

\def\Conf {\ensuremath{\mathrm{conf}}}

\def\IntF {\ensuremath{\mathrm{intf}}}

\def\CnsQu {\ensuremath{\Cns^{\text{q}}}}

\RequirePackage{xparse}
\NewDocumentCommand{\ActvC}{o}{%
  \ensuremath{\varphi_\bullet%
    \IfValueT{#1}{(#1)}%
  }%
}
\NewDocumentCommand{\CorrC}{o}{%
  \ensuremath{\varphi_\rightleftharpoons%
    \IfValueT{#1}{(#1)}%
  }%
}
\NewDocumentCommand{\TrgtC}{o}{%
  \ensuremath{\varphi_\Rightarrow%
    \IfValueT{#1}{(#1)}%
  }%
}

\NewDocumentCommand{\DeQuery}{m o}{%
\ensuremath{%
	\left\lbrack%
	#1 | \,%
	\IfValueTF{#2}{%
		#2%
	}{%
		(\CnsQu,\SuppMin,\ConfMin)
	}%
	\right\rbrack
}%
}

\def\AtLeTxt {AtLeastOne}

\def\AtMoTxt {AtMostOne}

\def\EndTxt {End}

\def\RespTxt {Response}

\def\ChnRespTxt {Chain\RespTxt}

\def\PrecTxt {Precedence}
\def\PrecTxtShort {Prec.}
\def\AltPrecTxt {Alternate\PrecTxt}
\def\AltPrecTxtShort {Alt.\PrecTxtShort}

\def\AtMoTmp {\ensuremath{\DclrSty{\AtMoTxt}}}

\def\EndTmp {\ensuremath{\DclrSty{\EndTxt}}}

\def\AltPrecTmp {\ensuremath{\DclrSty{\AltPrecTxt}}}

\newcommand{\AtLe}[1] {\ensuremath{\DclrSty{\AtLeTxt}\left(#1\right)}}

\newcommand{\AtMo}[1] {\ensuremath{\DclrSty{\AtMoTxt}\left(#1\right)}}

\newcommand{\End}[1] {\ensuremath{\DclrSty{\EndTxt}\left(#1\right)}}

\newcommand{\Resp}[2] {\ensuremath{\DclrSty{\RespTxt}\left(#1,#2\right)}}

\newcommand{\ChnResp}[2] {\ensuremath{\DclrSty{\ChnRespTxt}\left(#1,#2\right)}}

\newcommand{\Prec}[2] {\ensuremath{{\DclrSty{\PrecTxt}}\left(#1,#2\right)}}

\newcommand{\AltPrec}[2] {\ensuremath{\DclrSty{\AltPrecTxt}\left(#1,#2\right)}}
\newcommand{\AltPrecShort}[2] {\ensuremath{\DclrSty{\AltPrecTxtShort}\left(#1,#2\right)}}

\input{addon/macros-DECLARE-drawing}

\tikzstyle{truecolor}=[fill=green!30]
\tikzstyle{temptruecolor}=[fill=cyan!30]
\tikzstyle{falsecolor}=[fill=red!30]
\tikzstyle{tempfalsecolor}=[fill=orange!30]

\tikzstyle{truestate}=[state,truecolor,line width=2pt]
\tikzstyle{temptruestate}=[state,temptruecolor,very thick]
\tikzstyle{falsestate}=[state,falsecolor,very thick,densely dotted]
\tikzstyle{tempfalsestate}=[state,tempfalsecolor,very thick,densely dashed]

\tikzstyle{monitorstate}=[
          very thick,
          draw,
          inner sep=0pt]

\tikzstyle{truemonstate}=[
          monitorstate,
          truecolor,
          label=center:\permsat]
          
\tikzstyle{temptruemonstate}=[
          monitorstate,
          very thick,
          temptruecolor,
          label=center:\currsat]
          
\tikzstyle{falsemonstate}=[
          monitorstate,
          very thick,
          falsecolor,
          label=center:\permvio]
                    
\tikzstyle{tempfalsemonstate}=[
          monitorstate,
          very thick,
          tempfalsecolor,
          label=center:\currvio]

\tikzstyle{dot}=[
    circle,
    ultra thick,
    minimum width=1mm,
    minimum height=1mm,
    draw,
  ]

\input{addon/macros-math}
\def\Pn {\ensuremath{\mathcal{P}\hspace{-0.5ex}\mathcal{N}}}
\def\Wn {\ensuremath{\mathcal{W}\hspace{-0.5ex}\mathcal{N}}}
\DeclareDocumentCommand{\ReachGraph}{m o}
    {%
        \ensuremath{\IfValueTF{#2}{\mathcal{G}_{#1}^{#2}}{\mathcal{G}_{#1}}}
    }
\DeclareDocumentCommand{\ReachMarkings}{m o}
    {%
        \ensuremath{\IfValueTF{#2}{\mathcal{M}_{#1}^{#2}}{\mathcal{M}_{#1}}}
    }

\def\Places {\ensuremath{P}}
\def\Place {\ensuremath{p}}
\def\InPl {\ensuremath{{\blacktriangleright}}}
\def\OutPl {\ensuremath{\scriptstyle{\blacksquare}\displaystyle}}
\newcommand{\Pre}[1]{%
	{\begin{tikzpicture}[baseline=-0.6ex,>=stealth',x=0.75ex]%
			\draw[-Stealth] (0,0)--(1,0);%
	\end{tikzpicture}}%
	\ensuremath{#1}%
}
\newcommand{\Post}[1]{%
	\ensuremath{#1}%
	{\begin{tikzpicture}[baseline=-0.6ex,>=stealth',x=0.75ex]%
			\draw[-Stealth] (0,0)--(1,0);%
	\end{tikzpicture}}%
}

\def\Transitions {\ensuremath{T}}
\def\Transition {\ensuremath{t}}
\newcommand{\Tx}[1]{\ensuremath{\Transition{#1}}}

\def\FlowRel {\ensuremath{F}}

\def\Marking {\ensuremath{M}}

\DeclareDocumentCommand{\Fires}{m o m m}
    {%
        \ensuremath{\IfValueTF{#2}{(#2, #1) [#3\rangle (#2, #4)}{#1 [#3\rangle #4}}
    }
\def\PnFireSeq {\ensuremath{\sigma}}
\def\PnSetFireSeq {\ensuremath{\mathscr{L}}}

\newcounter{dummy} 
\newcounter{dummy1} 
\newcounter{dummy2}
\newcounter{dummy3} 
\newcounter{dummy4}
\newcounter{dummy5} 
\newcounter{dummy6}
\newcounter{dummy7}
\newcounter{dummy8}
\theoremseparator{.}
\theorembodyfont{\itshape}
\theoremsymbol{$\triangleleft$}
\newtheorem{theorem}[dummy]{Theorem}
\newtheorem{lemma}[dummy1]{Lemma}
\newtheorem{definition}[dummy2]{Definition}
\newtheorem{proposition}[dummy3]{Proposition} 
\newtheorem{observation}[dummy8]{Observation}
\newtheorem{corollary}[dummy7]{Corollary}

\theorembodyfont{\normalfont}
\newtheorem{example}[dummy4]{Example}
\newtheorem{remark}[dummy5]{Remark}

\theoremstyle{nonumberplain}
\theoremheaderfont{\itshape}
\theorembodyfont{\normalfont}

\theoremseparator{.}
\theoremsymbol{\ensuremath{\dashv}}
\newtheorem{proof}[dummy6]{Proof}

\qedsymbol{\ensuremath{\dashv}}

\renewcommand{\arraystretch}{1.5}
\renewcommand\tabcolsep{5pt}
\newcolumntype{d}{>{\columncolor{gray!10}}c}
\newcolumntype{m}{>{\columncolor{gray!10}}l}
\newfloatcommand{capbtabbox}{table}[][\FBwidth]
\newenvironment{iiilist}%
{\begin{inparaenum}[\itshape(i)\upshape]}%
{\end{inparaenum}}

\RequirePackage{xparse}
\NewDocumentEnvironment{AuthNote}{+o+o}{%
	\IfValueT{#2}{\marginnote{\scriptsize{#2}}}%
	\begin{scriptsize}%% Change the default size of the font
		\colorbox{gray}%
		{\color{white} Note\IfValueT{#1}{ (#1)}:}%
		\quad%
		\color{brown}%% Change the colour to brown
}{%
	\normalcolor%% Restore the normal colour
	\end{scriptsize}%% Restore the normal size
}

\RequirePackage{pifont}

\newcommand{\StarNum}[1] {% 
	\resizebox{2.25ex}{!}{%
	\ensuremath{%
		\begin{tikzpicture}[baseline=(star.south)] 
				\node [star, star point height=1ex, 
				minimum size=.5ex, inner sep=0.05ex, draw](star)
				at (1,1) {#1};
		\end{tikzpicture}}%
	}%
}
\def\StarOne   {\StarNum{1}}
\def\StarTwo   {\StarNum{2}}
\def\StarThree {\StarNum{3}}

\RequirePackage{lipsum}
\newcommand{\LipsumGray}[1][]{{\color{gray}\ifthenelse{\equal{#1}{}}{\lipsum}{\lipsum[#1]}}}

\RequirePackage{siunitx}
\newcolumntype{D}[1]{S[
	table-omit-exponent,
	round-mode=places,
	round-integer-to-decimal,
	round-precision={#1}]} % Rounds to the given number of decimals

\DeclareDocumentCommand{\crefalgln}{ o m }{%
	\IfNoValueF{#1}{\cref{#1}, }% if the first optional parameter is passed, print the reference to the algorithm
	{\hyperref[#2]{ln.~\ref{#2}}}}% Print the line number with reference
\DeclareDocumentCommand{\Crefalgln}{ o m }{%
	\IfNoValueF{#1}{\Cref{#1}, }% if the first optional parameter is passed, print the reference to the algorithm
	{\hyperref[#2]{line~\ref{#2}}}}% Print the line number with reference

\newacronym[\glslongpluralkey={Business Processes}]{bp}{BP}{Business Process}
\newacronym{bpi}{BPI}{Business Process Intelligence}
\newacronym{bpm}{BPM}{Business Process Management}
\newacronym{bpms}{BPMS}{Business Process Management System}
\newacronym{bpmn}{BPMN}{Business Process Model and Notation}
\newacronym{cmmn}{CMMN}{Case Management Model and Notation}
\newacronym{cpn}{CPN}{colored Petri net}
\newacronym{dcrg}{DCR~Graph}{Dynamic Condition Response Graph}
\newacronym{gsm}{GSM}{Guard Stage Milestone}
\newacronym{kpi}{KPI}{Key Performance Indicator}
\newacronym{ocbc}{OCBC}{Object-centric Behavioral Constraints}
\newacronym{soa}{SOA}{Service-Oriented Architecture}
\newacronym{pn}{PN}{Petri net}
\newacronym{wf}{WF}{workflow}
\newacronym{wfms}{WfMS}{Workflow Management System}
\newacronym{xes}{XES}{eXtensible Event Stream}
\newacronym{yawl}{YAWL}{Yet Another Workflow Language}
\newglossaryentry{task}{%
	name={task},description={the non-divisible, elementary activity}}
\def\paramx {\ensuremath{x}}
\def\paramy {\ensuremath{y}}

\def\lettera {\ensuremath{\textsl{a}}}
\def\letterb {\ensuremath{\textsl{b}}}
\def\letterc {\ensuremath{\textsl{c}}}
\def\letterd {\ensuremath{\textsl{d}}}
\def\lettere {\ensuremath{\textsl{e}}}
\def\letterf {\ensuremath{\textsl{f}}}
\def\letterg {\ensuremath{\textsl{g}}}

\def\letteri {\ensuremath{\textsl{i}}}
\def\letterj {\ensuremath{\textsl{j}}}

\def\lettero {\ensuremath{\textsl{o}}}

\def\letteru {\ensuremath{\textsl{u}}}
\def\letterv {\ensuremath{\textsl{v}}}
\def\letterw {\ensuremath{\textsl{w}}}

\def\letterA {\ensuremath{\textsl{A}}}

\newcommand{\Task}[1] {\ensuremath{\scalebox{0.85}{\normalfont\textsf{#1}}}}
\def\taska {\Task{a}}
\def\taskb {\Task{b}}
\def\taskc {\Task{c}}
\def\taskd {\Task{d}}
\def\taske {\Task{e}}
\def\taskf {\Task{f}}
\def\taskg {\Task{g}}

\def\taski {\Task{i}}
\def\taskj {\Task{j}}
\def\taskk {\Task{k}}

\def\taskm {\Task{m}}
\def\taskn {\Task{n}}

\def\taskp {\Task{p}}

\def\taskr {\Task{r}}

\def\tasku {\Task{u}}
\def\taskv {\Task{v}}
\def\taskw {\Task{w}}

\def\tasky {\Task{y}}
\def\taskz {\Task{z}}

\newglossaryentry{promod}{%
	name={process model},description={the model of a process}
}
\def\LogAlph {\ensuremath{\Sigma}}
\newglossaryentry{logalph}{
	name={log alphabet},description={the process alphabet of activities, as reflected in a log},%
	symbol={\LogAlph}}
\def\ProActvts {\ensuremath{\mathrm{Act}}}
\newglossaryentry{proalph}{
	name={process alphabet},description={the process alphabet of activities},%
	symbol={\ProActvts}}
\def\Evt {\ensuremath{e}}
\newglossaryentry{evt}{
	name={event},description={a record of an instantaneous fact during the process enactment},%
	symbol={\Evt}}
\def\Trc { \ensuremath{t} }

\newglossaryentry{trace}{
	name={trace},description={a sequence of \glsplural{evt}},%
	symbol={\Trc}}
\def\EvtLog {\ensuremath{L}}

\newglossaryentry{evtlog}{
	name={event log},description={a collection of \glstext{evttrace}s},%
	symbol={\EvtLog}}
\newglossaryentry{declare}{%
	name={\Declare},description={a declarative process modelling language and notation}}
\newglossaryentry{declaspec}{%
	name={declarative specification},description={a process specification, expressed by means of constraints},
	symbol={\DecSpec}
}

\newglossaryentry{mindeclamodel}{%
	name={discovered \glsentrytext{declaspec}},description={\glsentrydesc{declaspec}, discovered from an \glsentrytext{evtlog}},
	symbol={\MinDecSpec}
}
\newglossaryentry{minerful}{%
	name={\DclrSty{miner}ful},description={A declarative process discovery algorithm}}
\newacronym{mf}{Mf}{\gls{minerful}}
\newglossaryentry{minerfulVac}{%
	name={MINERful Vacuity Checker},description={\glsentrytext{minerful}} algorithm with semantical vacuity detection}
\newacronym{mfv}{Mf-Vchk}{\gls{minerfulVac}}
\newacronym{dmm}{DMM}{Declare Maps Miner}
\newglossaryentry{decmapmin}{%
	name={Declare Maps Miner},description={the declarative process discovery algorithm \glsentrytext{decmapmin}}}
\newacronym{dmm2}{DM2}{Declare Miner 2}
\newglossaryentry{decmapmin2}{%
	name={Declare Miner 2},description={improvement of \glsentrytext{decmapmin} algorithm}}
\newglossaryentry{janus}{%
	name={Janus},description={the declarative process discovery algorithm \glsentrytext{janus}}}
\newglossaryentry{subsum}{%
	name={subsumption},description={is subsumed by},%
	symbol={\Subsum}}
\newglossaryentry{relaxop}{%
	name={relaxation},description={relaxation operator, climbing the \glsentrytext{subsum} hierarchy}}
\newglossaryentry{actv}{%
	name={activation},description={the activation of a constraint}}
\newglossaryentry{activator}{name={activator},description={the event that signals the occurrence of the activation in the trace}}
\newglossaryentry{target}{%
	name={target},description={target}}
\newglossaryentry{cns}{%
	name={constraint},description={a temporal business process rule},
	symbol={\Cns}
}
\newglossaryentry{welldef}{%
	name={well-defined},description={of \glsentrytext{con}s for which a finite non-empty trace exists that complies with them}
}
\newglossaryentry{cnsset}{%
	name={constraint set},description={a set of \glsentrytext{con}s},
	symbol={\CnsSet}
}
\newglossaryentry{cnspar}{%
	name={parameter},description={a parameter of a \glsentrytext{con}},
}
\newglossaryentry{cnsarity}{%
	name={arity},description={number of parameters of a \glsentrytext{con}},
}
\newglossaryentry{exi}{
	name={existence},
	description={constrains single tasks}
}
\newglossaryentry{exicon}{
	name={\glsentrytext{exi} \glsentrytext{con}},
	description={constrains single tasks}
}
\newglossaryentry{posicon}{
	name={position \glsentrytext{con}},
	description={constrains the position of tasks}
}
\newglossaryentry{cardicon}{
	name={cardinality \glsentrytext{con}},
	description={limits the number of tasks}
}
\newglossaryentry{rela}{
	name={relation},
	description={constraint on pairs of tasks}
}
\newglossaryentry{relacon}{
	name={\glsentrytext{rela} \glsentrytext{con}},
	description={constraint on pairs of tasks}
}
\newglossaryentry{unirelacon}{
	name={unidirectional \glsentrytext{relacon}},
	description={constraint on pairs of tasks, out of which one is the activation, as the other is the target}
}
\newglossaryentry{unifwrelacon}{
	name={\glsentrytext{fw}-\glsentrytext{unirelacon}},
	description={constraint on pairs of tasks, having the first parameter as the activation, and the second one as the target}
}
\def\FwCns {\ensuremath{\mathit{fw}}}
\newglossaryentry{fw}{
	name={forward},
	description={forward constraint},
	symbol={\FwCns}
}

\newglossaryentry{unibwrelacon}{
	name={\glsentrytext{bw}-\glsentrytext{unirelacon}},
	description={constraint on pairs of tasks, having the second parameter as the activation, and the first one as the target}
}
\def\BwCns {\ensuremath{\mathit{bw}}}
\newglossaryentry{bw}{
	name={backward},
	description={backward constraint},
	symbol={\BwCns}
}

\newglossaryentry{corelacon}{
	name={coupling \glsentrytext{con}},
	description={constraint based on pairs of relation constraints}
}
\newglossaryentry{nega}{
	name={negative},
	description={of a constraint, that negates a coupling relation constraint}
}
\newglossaryentry{negacon}{
	name={\glsentrytext{nega} \glsentrytext{con}},
	description={constraint negating a coupling relation constraint}
}
\newglossaryentry{dectemp}{%
	name={template},description={the template of a \glsentrydesc{con}},
	symbol={\DecTemp}}
\newglossaryentry{cnstype}{%
	name={type},description={the type of a \glsentrydesc{cnstemp}}}
\newglossaryentry{cnsrep}{name={repertoire},description={the repertoire of declarative \glsentrytext{temp}s},
	symbol={\DecRepertoire}}
\newglossaryentry{cnsuniv}{name={\glsentrytext{con}s universe},description={the set of \glsentrytext{declare} \glsentrytext{temp}s over the process alphabet reflected in the log}}
\newglossaryentry{cnsinst}{%
	name={\glsentrytext{cnstemp} instantiation relation},description={the assignment relation instantiating \glsentrytext{cnstemp}s into \glsentrytext{con}s, namely assigning \glsentrytext{task}s to \glsentrytext{cnspar}s.},
	symbol={\CnsInstRelation}}
\newglossaryentry{cnsinterp}{
	name={interpretation function},description={function interpreting a \glsentrytext{declamodel}},
	symbol={\CnsInterpFun}}

\def\RelaConTemp {\ensuremath{\mathcal{R}}} % DEPRECATED
\newglossaryentry{relacontemp}{%
	name={relation template},description={the template of a relation \glsentrydesc{con}},
	symbol={\RelaConTemp}}

\def\ExiConTemp {\ensuremath{\mathcal{E}}}
\newglossaryentry{exicontemp}{%
	name={existence template},description={the template of an existence \glsentrydesc{con}},
		symbol={\ExiConTemp}}

\newglossaryentry{support}{%
	name={support},description={the support of a \glsentrydesc{con}},
	symbol={\Supp}}
\newglossaryentry{conf}{%
	name={confidence},description={the confidence level of a \glsentrydesc{con}},
	symbol={\Conf}}
\newglossaryentry{intf}{%
	name={interest factor},description={the interest factor of a \glsentrydesc{con}},
	symbol={\IntF}}
\newglossaryentry{cnseval}{
	name={evaluation},description={evaluation of a \glsentrytext{con} or a \glsentrytext{declamodel} over a \glsentrytext{evttrace} or an \glsentrytext{evtlog}},
	symbol={\CnsEvalFunctor}}

\newglossaryentry{fulfilment}{name={fulfilment},description={satisfaction of a constraint on a trace in which the activation occurs}}
\usetikzlibrary{DECLARE}
\def\Au {\ensuremath{\mathcal{A}}}
\def\AuAlph {\ensuremath{\Sigma}}
\def\AuSym {\ensuremath{\ell}}
\def\AuTrns {\ensuremath{\delta}}
\def\AuSt {\ensuremath{s}}
\def\AuStInit {\ensuremath{\AuSt_0}}
\def\AuStSet {\ensuremath{S}}
\def\AuStAcc {\ensuremath{\AuSt_\text{F}}}
\newacronym[symbol=\Au,longplural={finite state automata}]{fsa}{FSA}{finite state automaton}
\newacronym[symbol=\Au,longplural={deterministic finite-state automata}]{dfa}{DFSA}{deterministic finite-state automaton}
\newacronym[symbol=\Au,longplural={nondeterministic finite-state automata}]{nfa}{NFSA}{nondeterministic finite-state automaton}
\newglossaryentry{fsainit}{name={initial state},description={initial state of the automaton},
	symbol=\AuStInit}

\def\LanguageFunctor {\ensuremath{\mathscr{L}}}

\newcommand{\LanguageFunc}[1] {\ensuremath{\LanguageFunctor\!\left(#1\right)}}

\def\LTL {\ensuremath{\textsc{LTL}}}

\def\LDL {\ensuremath{\textsc{LDL}}}
\def\LDLf {\ensuremath{\textsc{LDL}_f}}
\def\LTLf {\ensuremath{\textsc{LTL}_f}}

\newacronym{po}{PO}{Partial Order}
\newacronym{tl}{TL}{Temporal Logic}  % Alessio's addition
\newacronym{ltl}{\LTL}{Linear Temporal Logic}
\newacronym{ldl}{\LDL}{Linear Dynamic Logic}
\newacronym{ldlf}{\LDLf}{Linear Dynamic Logic over Finite Traces}
\newacronym{fol}{FOL}{First Order Logic}
\newacronym{ltlf}{\LTLf}{Linear Temporal Logic on Finite Traces}
\newacronym{ltlp}{LTLp}{Linear-time Temporal Logic with Past}
\def\ltlpf {\ensuremath{\textrm{LTLp}_f}}
\newacronym{ltlpf}{\ltlpf}{Linear-time Temporal Logic with Past on Finite Traces}
\newacronym{mso}{MSO}{Monadic Second Order Logic}
\newacronym{rex}{RE}{regular expression}
\def\MultiSetFunctor {\ensuremath{\mathbb{M}}}
\newglossaryentry{multiset}{
	name={multi-set},description={a collection possibly containing multiple units of the same element},
	symbol={\MultiSetFunctor}}

\def\PowerSetFunctor {\ensuremath{\mathbb{P}}}
\newglossaryentry{powerset}{
	name={power-set},description={the collection of sets generated by all combinations without repetition of elements in a set},
	symbol={\PowerSetFunctor}}

\newcommand{\ltrue}{\ensuremath{\mathbf{true}}}
\newcommand{\lfalse}{\ensuremath{\mathbf{false}}}
\newcommand{\limply}{\ensuremath{\to}}

\newcommand{\lmodel}{\ensuremath{\vDash}}
\newcommand{\lnmdel}{\ensuremath{\nvDash}}
\def\ltllnext {\ensuremath{\;\mathop{\mathrm{\mathbf{X}}}\;}}
\def\ltllevtly {\ensuremath{\;\mathop{\mathrm{\mathbf{F}}}\;}}
\def\ltllalws {\ensuremath{\;\mathop{\mathrm{\mathbf{G}}}\;}}
\def\ltlluntil {\ensuremath{\;\mathop{\mathrm{\mathbf{U}}}\;}}

\def\ltllyday {\ensuremath{\;\mathop{\mathrm{\mathbf{Y}}}\;}}
\def\ltllonce {\ensuremath{\;\mathop{\mathrm{\mathbf{O}}}\;}}

\def\ltllalws {\ensuremath{\;\mathop{\mathrm{\mathbf{G}}}\;}}
\def\ltllsince {\ensuremath{\;\mathop{\mathrm{\mathbf{S}}}\;}}
\def\ltlnext{\ltllnext}
\def\ltlevtly{\ltllevtly}
\def\ltlalws{\ltllalws}
\def\ltluntil{\ltlluntil}

\def\ltlyday{\ltllyday}
\def\ltlonce{\ltllonce} 
 
\def\ltlalws{\ltllalws} 
\def\ltlsince{\ltllsince} 

\def\lstart {\ensuremath{\mathbf{start}}}

\def\linstant {\ensuremath{i}}
\def\ltrace {\ensuremath{\pi}}
\def\ltrlen {\ensuremath{n}}
\def\currvio {\ensuremath{{\textsc{c}}\!\bot}}
\def\permvio {\ensuremath{{\textsc{p}}\!\bot}}
\def\currsat {\ensuremath{{\textsc{c}}\!\top}}
\def\permsat {\ensuremath{{\textsc{p}}\!\top}}

\NewDocumentCommand{\valuation}{m o}{%
	\ensuremath{%
		\left\llbracket%
		\IfValueTF{#2}{%
			#2 \mapsfrom #1%
		}{%
			#1
		}%
		\right\rrbracket
	}%
}

\newcommand{\FormulaOf}[1]{\ensuremath{\varphi_{#1}}}

\newglossaryentry{satis}{name={satisfaction},description={evaluation as true of a formula on a structure}}
\newglossaryentry{verif}{name={verification},description={evaluation of a formula on a structure}}
\def\post#1{\ensuremath{{#1}\kern-.05ex\bullet}}

\usepackage{authblk} % for author/institute formatting

\newcommand{\keywords}[1]{%
  \vspace{0.5em}
  \noindent\textbf{Keywords:} #1
}

\title{From Sound Workflow Nets \\ to LTL$_f$ Declarative Specifications \\ by Casting Three Spells}

\author[1]{\small Luca Barbaro}
\author[2]{Giovanni Varricchione}
\author[3]{Marco Montali} 
\author[2]{Claudio Di Ciccio}

\affil[1]{\small\textit{Sapienza University of Rome, Italy\\
\href{mailto:luca.barbaro@uniroma1.it}{luca.barbaro@uniroma1.it}}}

\affil[2]{\small\textit{Utrecht University, Netherlands\\
\href{mailto:g.varricchione@uu.nl}{g.varricchione@uu.nl}, 
\href{mailto:c.diciccio@uu.nl}{c.diciccio@uu.nl}}}

\affil[3]{\small\textit{Free University of Bozen-Bolzano, Italy\\
\href{mailto:montali@inf.unibz.it}{montali@inf.unibz.it}}}

\date{} % remove if you want today's date
\begin{document}

\maketitle

\vspace*{-4em}

\begin{abstract}
In process management, effective behavior modeling is essential for understanding execution dynamics and identifying potential issues. 
Two complementary paradigms have emerged in the pursuit of this objective: the imperative approach, representing all allowed runs of a system in a graph-based model, and the declarative one, specifying the rules that a run must not violate in a constraint-based specification.
Extensive studies have been conducted on the synergy and comparisons of the two paradigms. To date, though, whether a declarative specification could be systematically derived from an imperative model such that the original behavior was fully preserved (and if so, how) remained an unanswered question. 
In this paper, we propose a three-fold contribution.
(1)~We introduce a systematic approach to synthesize declarative process specifications from safe and sound Workflow nets.
(2)~We prove behavioral equivalence of the input net with the output specification, alongside related guarantees.
(3)~We experimentally demonstrate the scalability and compactness of our encoding through tests conducted with synthetic and real-world testbeds.

\end{abstract}

 \keywords{%
 	Process modeling,
 	\and
 	Petri nets,
        \and
        Linear-time Temporal Logic on finite traces,
 	\and
 	{\Declare}%
}

\section{Introduction}
\label{sec:intro}
The act of modeling a process is a key element in a multitude of domains, including business process management~\cite{DBLP:books/sp/DumasRMR18}, and is specifically tailored to meet the specific requirements and objectives of the individual application scenarios.
Two fundamental, complementary paradigms cover the spectrum of modeling: the imperative (e.g., Worflow nets~\cite{DBLP:journals/jcsc/Aalst98} and their derived industrial standard BPMN~\cite{DBLP:books/sp/DumasRMR18}) and the declarative (e.g., \Declare~\cite{Pesic/2008:ConstraintbasedWorkflow} and DCR Graphs~\cite{DBLP:journals/corr/abs-1110-4161}).
Generally, the former class offers the opportunity to explicitly capture the set of actions available at each reachable state of the process, from start to end.
However, such models often show limitations when it comes to capture flexibility in execution, since the possible runs highly vary and their graph-based structure gets cluttered. 
To compactly represent that variability, declarative specifications depict the rules that govern the behavior of every instance, leaving the allowed sequences implicit as long as none of those rules is violated.
 % provide an advantage when operating in intricate and mutable environments, which stems from the effectiveness to describe the behavior in a more compact way. 

Research has acknowledged that none of the available representations would be superior in all cases, as imperative and declarative approaches are apt to different comprehension tasks~\cite{DBLP:conf/bpm/PichlerWZPMR11}.
The ability to translate one representation to the other while preserving behavioral equivalence would allow the comparison and selection of the most suitable one.
The first work in this direction is~\cite{DBLP:conf/simpda/PrescherCM14}, where a systematic procedure is proposed to turn a declarative specification into an imperative model. Other endeavors followed to close the circle by providing an approximate solution to the inverse path (i.e., from an imperative model to a declarative specification), resorting on re-discovery over simulations~\cite{DBLP:journals/corr/abs-2407-02336},
state space exploration~\cite{DBLP:conf/bpm/RochaZA24}, or behavioral comparison~\cite{DBLP:conf/bpm/BergmannRK23}. 

The goal of this work is to close the existing gap of this procedural-to-declarative direction. To this end, we show how to encode a safe and sound Workflow net~\cite{DBLP:journals/jcsc/Aalst98}
into a behaviorally equivalent {\Declare} specification~\cite{DiCiccio.Montali/PMH2022:DeclarativeProcessMining}. %expressed in \LTLf~\cite{DeGiacomo.Vardi/IJCAI2013:LDLf} with join-semilattice interpretations~\cite{Chechik.etal/ACMTOSEM2003:MultiValuedSymbolicModelChecking}
%
%. This class of nets includes well-structured and free-choice nets, namely formalisms at the core of widespread classes of process models such as the well-formed BPMN diagrams~\cite{Ouyang.etal/TOSEM2009:BusinessProcessModels2ProcessOrientedSwSystems}. 
%
The three spells mentioned in the title of this paper refer to the fact that we only employ three parametric constraint types (\emph{templates}) in the {\Declare} repertoire.
Importantly, the encoding is obtained in one pass and modularly over the net, preserving runs and choice points without incurring the state space explosion caused by concurrency unfolding.
A byproduct of the encoding is that a safe and sound Workflow net induces a star-free regular language when considering transitions of the former as the alphabet of the latter.
This strengthens the well-known fact that languages induced by sound Workflow nets are regular.
Then, we evaluate the scalability of our approach by experimentally testing our proof-of-concept implementation against synthetic and real-world testbeds. Also, we show a downstream reasoning task on process diagnostics with public benchmarks. %This result facilitates comprehension and integration between models, not least providing helpful support for behavior verification using process diagnostics tools.

The remainder of the paper is organized as follows. \Cref{sec:background} provides an overview of the background knowledge our research is built upon. We describe our algorithm and formally discuss its correctness and complexity in \Cref{sec:algo}. \Cref{sec:evaluation} evaluates our implementation to demonstrate the feasibility of our approach. Finally, we conclude by discussing related works in \Cref{sec:related} and outlining future research directions in \Cref{sec:conclusion}.

\section{Background}
\label{sec:background}

% Considering that {\Declare} templates have been originally defined starting from a catalogue of \gls{ltl} patterns \cite{Dwyer.etal/ICSE1999:PatternsinPropertySpecificationsFiniteStateVerification}, it is not surprising that temporal logics have been used to characterize the semantics of {\Declare} since the very beginning. However, the fact that {\Declare} specifications are interpreted over finite-length executions calls for the use of \gls{ltlf}~\cite{DeGiacomo.Vardi/IJCAI2013:LDLf}. This indeed leads to a setting that is radically different, both semantically and algorithmically, from the traditional one where formulae are interpreted using {\LTL} over infinite, recurring behaviors \cite{DeGiacomo.etal/AAAI2014:ReasoningLTLFinite}. 

% A complete formalization of {\Declare} templates, also including an alternative formalization using a logic programming-based approach, can be found in \cite{Montali/2010:SpecificationandVerification}. It was later refined in~\cite{DeGiacomo.etal/AAAI2014:ReasoningLTLFinite}. In his PhD thesis, Di~Ciccio was the first to provide a semantics based on regular expressions \cite{DiCiccio/2013:MiningArtfulProcesses}. These two themes were later unified in \cite{DDMM22}, leading to a richer framework that is able to declaratively capture constraints and metaconstraints, that is, constraints predicating over the possible/certain satisfaction and violation of other constraints.
%
In this section, we formally describe the foundational pillars our work is built upon. %as well as on the automata-theoretic characterization for this logic. We then use this framework to formalize {\Declare} and reason automatically on {\Declare} specifications.
% \todo{To be entirely rewritten.}
%After an outlook of 
% 
%
%
%
%Here we provide the formal notions that underpin the execution model of a declarative process specification, starting with the temporal logics with which its semantics is expressed, and then briefly describing the execution model as finite-state automata.
%
%
%
%
%Declarative specification languages provide a standard \emph{repertoire} of \emph{templates}, i.e., constraints parametrised over activities. The major benefit of using templates is that analysts do not have to be aware of the underlying logic-based formalisation to understand the models. They work with the graphical or textual representation of templates, while the underlying formulae remain hidden.
%
%Usually, the formal underpinning for such intuitive notions is provided by temporal logics, whose models are indeed traces. In particular, formulae of the logic are used to capture constraints, and logical consequence to unambiguously define when a trace satisfies a constraint and is compliant with a declarative process model \cite{Mon10}.

\subsection{Linear Temporal Logic on finite traces}
\label{sec:declare:ltl}
% !TeX program = pdflatex
% !BIB program = bibtex
% !TeX spellcheck = en_UK
% !TeX encoding = utf8
% !TeX root = ../main.tex
%
%The most widely adopted logic for declarative process modelling is \gls{ltlf}~\cite{DeGiacomo.Vardi/IJCAI2013:LDLf}. This logic is at the basis of concrete modeling languages such as \Declare. % (see~\cref{sec:declare}).
%\todo[inline]{Warning: the full text is taken as-is from \cite{DiCiccio.Montali/PMH2022:DeclarativeProcessMining}. We should use letter-symbols in place of arcane glyphs for \LTLf syntax.
%
%Giovanni: introduciamo tutta LTL$_f$ o soltanto il frammento che ci serve (Past LTL + Globally)?}

{\LTLf} has the same syntax of \LTL~\cite{Pnueli/FOCS1977:LTL}, but is interpreted on finite traces.
Here, we consider the {\LTL} dialect including past modalities~\cite{Lichtenstein.etal/LogicsofPrograms1985:TheGloryofthePast} for declarative process specifications as in~\cite{Cecconi.etal/BPM2018:Janus}.
From now on, we fix a finite set $\LogAlph$ representing an alphabet of propositional symbols.
%\todo{\textbf{IMPORTANT NOTE}: Since we will use transitions and not activity labels to express our constraints, we cannot call the symbols in the alphabet ``activities'', or we contradict ourselves.}
A (finite) \emph{trace} $\ltrace = \langle a_1,\ldots,a_n \rangle \in \LogAlph$ is a finite sequence of symbols of length $|\ltrace|=n$ (with $n \in \mathbb{N}$), where the occurrence of symbol $a_i$ at instant $i$ of the trace represents an \emph{event} that witnesses $a_i$ at instant $i$ ---we write $\ltrace(i) = a_i$. Notice that \emph{at each instant we assume that one and only one symbol occurs}. Using standard notation from regular expressions, $\LogAlph^*$ denotes the overall set of finite traces derived from events belonging to $\LogAlph$.

\begin{definition}[{\LTLf}]
\textbf{(Syntax)}
\label{def:ltlf-syntax}
%Let $\LogAlph \supseteq \{ \taska \}$ be an alphabet of propositional symbols closed under the boolean connectives.
Well-formed {\LTLf} formulae are built from a finite non-empty alphabet of symbols $\LogAlph \ni \lettera$, 
% \supseteq \{ \taska \}$, 
 the unary temporal operators {\ltlnext} (``\emph{next}'') and {\ltlyday} (``\emph{yesterday}''), and the binary temporal operators {\ltluntil} (``\emph{until}'') and {\ltlsince} (``\emph{since}'') as follows:
$$
\varphi \Coloneqq
\lettera 
\mid
(\lnot\varphi) 
\mid
(\varphi_1 \land \varphi_2)
\mid
( \ltlnext \varphi)
\mid
(\varphi_1 \ltluntil \varphi_2)
\mid
(\ltlyday\varphi)
\mid
(\varphi_1 \ltlsince \varphi_2).
$$
%\end{definition}
%%
%\begin{definition}[Semantics of {\LTLf}, satisfaction, validity, entailment]
\textbf{(Semantics)}
\label{def:ltlf-semantics}
An {\LTLf} formula $\varphi$ is inductively \emph{satisfied} in some instant {\linstant} (with $ 1 \leq \linstant \leq \ltrlen $) of a trace {\ltrace}
of length $\ltrlen \in \mathbb{N}$, written $\ltrace, \linstant \lmodel \varphi$, if the following holds:\\
\begin{itemize*}[label={},itemjoin=\hspace{0.5ex}]
%	\item $ \ltrace, \linstant \lmodel \ltrue $; $ \ltrace, \linstant \lnmdel \lfalse $;
	\item
	$ \ltrace, \linstant \lmodel \lettera $ iff $ \ltrace(\linstant) $ is assigned with $ \lettera $;
	\quad \item 
        $ \ltrace, \linstant \lmodel \lnot\varphi $ iff $ \ltrace, \linstant \lnmdel \varphi $;
	\quad \item
	$ \ltrace, \linstant \lmodel \varphi_1\land\varphi_2 $ iff $ \ltrace, \linstant \lmodel \varphi_1 $ and $ \ltrace, \linstant \lmodel \varphi_2 $; 
	%			\item $ \ltrace, \linstant \lmodel \ltlnext\varphi $ iff $ i<\ltrlen $ and $ \ltrace, \linstant+1 \lmodel \varphi $
	\\ \item $ \ltrace, \linstant \lmodel \ltlnext\varphi $ iff $ i < \ltrlen $ and $ \ltrace, \linstant+1 \lmodel \varphi $; 
	\quad \item $ \ltrace, \linstant \lmodel \ltlyday\varphi $ iff $ \linstant>1 $ and $ \ltrace, \linstant-1 \lmodel \varphi $; 
	\\ \item $ \ltrace, \linstant \lmodel \varphi_1\ltluntil\varphi_2 $ iff there exists $i \leq j \leq n$ such that $ \ltrace,j \lmodel\varphi_2 $, and $ \ltrace, k \lmodel \varphi_1 $ for all $k$ s.t.\ $ {\linstant\leq k<j} $; 
	\\ \item $ \ltrace, \linstant \lmodel \varphi_1\ltlsince\varphi_2 $ iff there exists $1 \leq j \leq i$ such that $ \ltrace, j \lmodel\varphi_2 $, and $ \ltrace, k \lmodel \varphi_1 $ for all $k$ s.t.\ $ {j < k \leq \linstant} $. 
\end{itemize*} %\medskip

A formula $\varphi$ \emph{is satisfied by} a trace $\ltrace$, %(equivalently, $\ltrace$ \emph{satisfies} $\varphi$), 
written $\ltrace\lmodel {\varphi}$, iff $\ltrace, 1 \lmodel {\varphi}$. 
%A formula $\varphi$ is:
%\begin{inparaenum}[\itshape (i)] 
%  \item \emph{satisfiable} if it has a satisfying trace from $\LogAlph^*$;
%  \item \emph{valid} if every trace in $\LogAlph^*$ satisfies it. 
%\end{inparaenum}
%A formula $\varphi_1$ \emph{entails} formula $\varphi_2$, written $\varphi_1 \models \varphi_2$, if, 	for every trace $\Trc$ of length $n \in \mathbb{N}$ and every $i$ s.t.\ $1 \leq i \leq n$, if $t,i \models \varphi$ then $t,i \models \psi$.
\end{definition}
% Since {\LTLf} is closed under negation, it is easy to see that a formula $\varphi$ is valid if and only if $\neg \varphi$ is unsatisfiable.
% It is worth noting that, in {\LTLf}, the next operator is interpreted as the so-called \emph{strong} next: $\ltlnext\varphi$ requires that the next instant exists within the trace, and that at such next instant $\varphi$ holds. This has an important consequence: differently from \LTL, in {\LTLf} formula $\neg \ltlnext \varphi$ is \emph{not} equivalent to $ \ltlnext \neg \varphi$. This is because $\neg \ltlnext \varphi$ is true in an instant of a finite trace either when that instant has no successor, or the next instant exists and in such a next instant $\varphi$ does not hold. More on this can be found in \cite{DeGiacomo.etal/AAAI2014:ReasoningLTLFinite}.
\noindent
From the basic operators above, the following can be derived: 
\begin{itemize*}[label={},itemjoin=\hspace{0.5ex}]
%\begin{compactitem}
	\item Classical boolean abbreviations $ \ltrue, \lfalse, \lor, \limply $;
%	\item Constant $\lend \equiv \lnot \ltlnext \ltrue $, denoting the last instant of a trace;
%	\item Constant $\lstart \equiv \lnot \ltlyday \ltrue $, denoting the first instant of a trace;
	\item $ \ltlevtly \varphi \equiv \ltrue \ltluntil \varphi $ indicating that $ \varphi $ eventually holds true in the trace (``\emph{eventually}''); % (hence, before or at {\lend}); %the last instant, there is an instant in which 
%	\item $ \varphi_1 \ltlwntil \varphi_2 \equiv (\varphi_1 \ltluntil \varphi_2) \lor \ltlalws \varphi_1$, which relaxes $\ltluntil$ as $\varphi_2$ may never hold true; %the last instant, there is an instant in which 
%	\item $ \ltlonce \varphi \equiv \ltrue \ltlsince \varphi $ indicating that $ \varphi $ holds true at some previous instant (``\emph{once}''); % (i.e., after {\lstart} in the trace); %the initial instant \lstart, there is a time where $ \varphi $ is $ \ltrue $;
	\item $ \ltlalws \varphi \equiv \lnot \ltlevtly \lnot\varphi $ indicating that $ \varphi $ holds true from the current on (``\emph{always}'').%; %the end of the trace \lend, $ \varphi $ is always \ltrue;
%	\item $ \ltlhist \varphi \equiv \lnot\ltlonce\lnot\varphi $ indicating that $ \varphi $ holds true from the start of the trace to the current instant. %, from the beginning of the trace \lstart to that instant, $ \varphi $ is always \ltrue.
%\end{compactitem}
\end{itemize*}
%	We remark that, w.l.o.g., we consider here the non-strict semantics of $\ltluntil$ and $\ltlsince$~\cite{hodkinson_separation-past_2005}.
%
%
%
%\begin{example}
	
    As an example, let 
	$\ltrace = \langle \lettera, \letterb, \letterc, \lettere, \letterf, \letterg, \letteru, \letterv \rangle$
	be a trace and $\varphi_1$, $\varphi_2$ and $\varphi_3$ three {\LTLf} formulae defined as follows:
	\begin{inparadesc}
		\item[$\varphi_1$]$\doteq \letterf$;
		\item[$\varphi_2$]$\doteq \ltlyday \letterc$;
		\item[$\varphi_3$]$\doteq \ltlalws ( \letterf \limply \ltlyday \lettere ) $.
	\end{inparadesc}
	We have that
	\begin{inparaitem}[]
		$\ltrace, 1 \lnmdel \varphi_1$ whereas $\ltrace, 5 \lmodel \varphi_1$;
		$\ltrace, 1 \lnmdel \varphi_2$ while $\ltrace, 4 \lmodel \varphi_2$;
		$\ltrace, 1 \lmodel \varphi_3$ (hence, $\ltrace \models \varphi_3$); in fact, $\ltrace, \linstant \lmodel \varphi_3$ for any instant $1 \leq \linstant \leq |\ltrace|$.
	\end{inparaitem}
%\end{example}
%
%Typically, when the instant \linstant is omitted, it is assigned with $1$, i.e., the first event in the trace.

\subsection{Finite State Automata}
\label{sec:fsa}
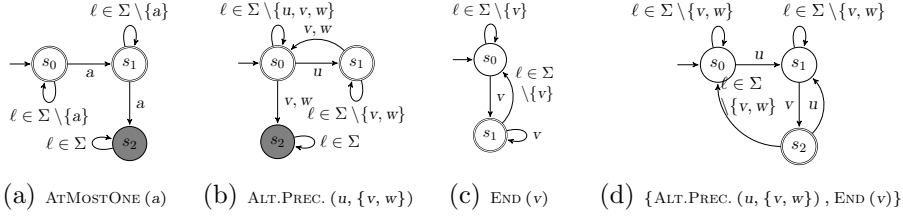
\begin{figure}[tb]%
	\def\maxheight{0.08\paperheight}
	\def\maxwidth{0.9\textwidth}
	\centering%
	\begin{subfigure}{0.22\textwidth}%
		\vbox to \maxheight {%
			\centering\maxsizebox{\maxwidth}{\maxheight}{%
				\def\taskASymbol{\lettera}
\begin{tikzpicture}[->, >=stealth', shorten >=1pt, auto, bend angle=45, initial text = {}]
  \tikzstyle{every state}=[minimum size=1em, align=center]

  \node[state, initial, accepting]    (0) {$\AuSt_0$};
  \node[state, right=of 0, accepting] (1) {$\AuSt_1$};
  \node[state, below=of 1, fill=gray] (2) {$\AuSt_2$};

  \path
  (0) edge [loop below] node []      {$\AuSym \in \AuAlph \setminus\! \left\lbrace \taskASymbol \right\rbrace$} (0)
  (0) edge [          ] node [below] {$\taskASymbol$}                                                           (1)
  (2) edge [loop left ] node []      {$\AuSym \in \AuAlph$}                                                     (2)
  (1) edge [          ] node [right] {$\taskASymbol$}                                                           (2)
  (1) edge [loop above] node         {$\AuSym \in \AuAlph \setminus\! \left\lbrace \taskASymbol \right\rbrace$} (1)
  ;  
\end{tikzpicture}%
			}%
		}%
		%		\caption{$\ResEx{\lettera}{\letterb}$}%
		\caption{\tiny{\AtMo{\lettera}}}
		\label{fig:constraint:automata:atmoone:inst}%
	\end{subfigure}%
	\hfill
	\begin{subfigure}{0.25\textwidth}%
		\vbox to \maxheight {%
			\centering\maxsizebox{\maxwidth}{\maxheight}{%
				\def\taskASymbol{\letteru}
\def\taskBSymbol{\letterv,\letterw}
\begin{tikzpicture}[->, >=stealth', shorten >=1pt, auto, bend angle=45, initial text = {}]
	\tikzstyle{every state}=[minimum size=1em, align=center]
	
	\node[state, initial, accepting]    (0) {$\AuSt_0$};
	\node[state, right=of 0, accepting] (1) {$\AuSt_1$};
	\node[state, below=of 0, fill=gray] (2) {$\AuSt_2$};
	
	\path
	(0) edge [loop above] node []      {$\AuSym \in \AuAlph \setminus\! \left\lbrace \taskASymbol, \taskBSymbol \right\rbrace$} (0)
	(0) edge [          ] node [below] {$\taskASymbol$}                                                            (1)
	(0) edge [          ] node [] {$\taskBSymbol$}                                                                 (2)
	(2) edge [loop right] node [] {$\AuSym \in \AuAlph$}                                                      (2)
	(1) edge [bend right] node [above=-2pt] {$\taskBSymbol$}                                                            (0)
	(1) edge [loop below] node         {$\AuSym \in \AuAlph \setminus\! \left\lbrace \taskBSymbol \right\rbrace$}        (1)
	;  
\end{tikzpicture}%
			}%
		}%
		\caption{\tiny{\AltPrecShort{\letteru}{\{\letterv,\letterw\}}}}
		\label{fig:constraint:automata:altprec:inst}%
	\end{subfigure}%
	\hfill
	\begin{subfigure}{0.15\textwidth}%
		\vbox to \maxheight {%
			\centering\maxsizebox{\maxwidth}{\maxheight}{%
				\def\taskASymbol{\letterv}
\begin{tikzpicture}[->, >=stealth', shorten >=1pt, auto, bend angle=45, initial text = {}]
	\tikzstyle{every state}=[minimum size=1em, align=center]
	
	\node[state, initial]               (0) {$\AuSt_0$};
	\node[state, below=of 0, accepting] (1) {$\AuSt_1$};
	
	\path
	(0) edge [loop above] node [] {$\AuSym \in \AuAlph \setminus\! \left\lbrace \taskASymbol \right\rbrace$} (0)
	(0) edge [          ] node [] {$\taskASymbol$}                                                                         (1)
	(1) edge [loop right] node [] {$\taskASymbol$}                                                                   (1)
	(1) edge [bend right] node [right,label={[align=right,xshift=1em,yshift=-1em]$\AuSym \in \AuAlph $ \\ $\setminus\! \left\lbrace \taskASymbol \right\rbrace$}] {}                                                                    (0)
	;  
\end{tikzpicture}%
			}%
		}%
		\caption{\tiny{\End{\letterv}}}%
		\label{fig:constraint:automata:end:inst}%
	\end{subfigure}%
	\hfill
	\begin{subfigure}{0.38\textwidth}%
		\vbox to \maxheight {%
			\centering\maxsizebox{\maxwidth}{\maxheight}{%
				\def\taskASymbol{\letteru}
\def\taskBSymbol{\letterv}
\def\taskCSymbol{\letterw}
\begin{tikzpicture}[->, >=stealth', shorten >=1pt, auto, bend angle=45, initial text = {}]
	\tikzstyle{every state}=[minimum size=1em, align=center]
	
	\node[state, initial]               (0) {$\AuSt_0$};
	\node[state, right=of 0]            (1) {$\AuSt_1$};
	\node[state, below=of 1, accepting] (2) {$\AuSt_2$};
%	\node[state, below=of 0, fill=gray] (3) {$\AuSt_3$};
	
	\path
	(0) edge [loop above] node [above,xshift=-2em]      {$\AuSym \in \AuAlph \setminus\! \left\lbrace \taskBSymbol, \taskCSymbol \right\rbrace$} (0)
	(0) edge [          ] node [above] {$\taskASymbol$}                                                            
	(1)
%	(0) edge [bend right] node [right] {$\taskBSymbol$}                                                                 
%	(3)
%	(3) edge [loop left ] node [label={[above=-2pt]{$\AuSym \in \AuAlph$}}] {}                                                      
%	(3)
	(1) edge [          ] node [left] {$\taskBSymbol$}                                                            
	(2)
	(2) edge [bend right] node [] {$\taskASymbol$}                                                            
	(1)
	(2) edge [bend left] node [above right,xshift=4pt,yshift=-4pt,label={[align=left]$\AuSym \in \AuAlph $ \\ $\setminus\! \left\lbrace \taskBSymbol, \taskCSymbol \right\rbrace$}] {}
	(0)
	(1) edge [loop above] node [above,xshift=2em] {$\AuSym \in \AuAlph \setminus\! \left\lbrace \taskBSymbol, \taskCSymbol \right\rbrace$}        (1)
	;  
\end{tikzpicture}%
			}%
		}%
		\caption{\tiny{$\{ \AltPrecShort{\letteru}{\{\letterv,\letterw\}}, \End{\letterv} \}$}}%
		\label{fig:constraint:automata:product:inst}%
	\end{subfigure}%
	\caption[Example FSAs]{Example \acrshortpl{fsa} of {\Declare} constraints. The \acrshort{fsa} in \cref{fig:constraint:automata:product:inst} is trimmed.}
	\label{fig:constraint:automata}
\end{figure}
Every {\LTLf} formula can be encoded into a \emph{deterministic finite state automaton}~\cite{DeGiacomo.etal/AAAI2014:ReasoningLTLFinite}.

\begin{definition}[Finite State Automaton]
    A (deterministic) \acrfull{fsa} is a tuple $\Au = \left( \AuStSet, \AuStInit, \AuStAcc, \AuAlph, \AuTrns \right) $, where $\AuStSet$ is a finite set of states, $\AuStInit \in \AuStSet$ is the initial state, $\AuStAcc \subseteq \AuStSet$ is the set of accepting states, $\AuAlph$ is the input alphabet of the automaton, and $\AuTrns: \AuStSet \times \AuAlph \to \AuStSet$ is the state transition function.
\end{definition}
\noindent
\Cref{fig:constraint:automata} depicts three \glspl{fsa}.
An \gls{fsa} reads in input sequences of symbols (``\emph{string}'') of its input alphabet. It starts in its initial state $\AuStInit$ and updates the state after having read each symbol via the state transition function $\AuTrns$. We say that a \gls{fsa} \emph{accepts} %or \emph{recognizes} 
a string if after reading it is in one of its accepting states (i.e., a state in $\AuStAcc$), and otherwise we say that it \emph{rejects} that string. The set of strings accepted by an \gls{fsa} $\Au$ is called the \emph{language} of $\Au$.

\begin{definition}[Bisimilarity]
Two {\gls{fsa}}s $\Au = (\AuStSet, \AuStInit, \AuStAcc, \AuAlph, \AuTrns) $ and $\Au' = ( \AuStSet', \AuStInit', \AuStAcc', \AuAlph, \AuTrns') $ are \emph{bisimilar} if and only if there exists a relation $\sim \subset \AuStSet \times \AuStSet'$ such that the following hold: 
\begin{itemize*}[label={},itemjoin=\quad]
    \item $(\AuStInit, \AuStInit') \in \sim $;
    \item if $(\AuSt, \AuSt') \in \sim$, then $(\delta(\AuSt, \AuSym), \delta'(\AuSt', \AuSym)) \in \sim$ for any $\AuSym \in \AuAlph$;
    \item if $(\AuSt, \AuSt') \in \sim$, then $\AuSt \in \AuStAcc$ if and only if $\AuSt' \in \AuStAcc'$.
\end{itemize*}
\end{definition}

\begin{observation}\label{obs:deterministic:fsa:bisimulation:language:equivalence}
In the case of \glspl{fsa}, bisimilarity coincides with \emph{language equivalence}, i.e., two \gls{fsa}s are bisimilar if and only if the sets of strings that they accept are equal~\cite{Hopcroft.etal/2006:IntroductiontoAutomataTheoryLanguagesandComputation}. %\footnote{This follows from the fact that two deterministic labelled transition systems are bisimilar if and only if they }.
\end{observation}
% Note that condition (2) is equivalent to requiring that the underlying labelled state transition systems of the \gls{fsa}s (i.e., the automata but without initial and accepting states) are bisimilar as well. 

A direct approach that builds a non-deterministic \gls{fsa} $\Au_\varphi$ accepting all and only the traces that satisfy a given {\LTLf} formula $\varphi$ is presented in \cite{DeGiacomo.etal/AAAI2014:ReasoningLTLFinite}.
%DDMM22}. 
We make two further observations from~\cite{Hopcroft.etal/2006:IntroductiontoAutomataTheoryLanguagesandComputation}: 
\begin{iiilist}\item the so-obtained \glspl{fsa} can be determinized, minimized, and trimmed using standard techniques without modifying the accepted language, and \item given any two \glspl{fsa} $\Au_\varphi$ and $\Au_{\varphi'}$, their \emph{product} $\Au_\varphi \times \Au_{\varphi'}$ recognizes all and only the traces of $\varphi \wedge \varphi'$.\end{iiilist}

\subsection{{\LTLf}-based declarative specifications}
\label{sec:declare:semantics}
% !TeX program = pdflatex
% !BIB program = bibtex
% !TeX spellcheck = en_UK
% !TeX encoding = utf8
% !TeX root = ../main.tex
%
\begin{table}[tb]
	\centering
	\caption{Semantics of some {\Declare} constraint templates}
	\label{tab:declare:semantics}
	\begin{adjustbox}{width=1\columnwidth,center=\columnwidth}
		\begin{scriptsize}
			\centering
			%!TEX root = ../main.tex
%
%\renewcommand{\arraystretch}{1.6}
\begin{tabular}{ l l l }
	\toprule
	\textbf{Template} & 
	\textbf{\LTLf expression~\cite{DeGiacomo.Vardi/IJCAI2013:LDLf,Cecconi.etal/BPM2018:Janus}} &
	\textbf{Description}
\\

\midrule
$\AtMo{\paramx}$ &
 $ \ltlalws (\paramx \limply \lnot \ltlnext\ltlevtly\paramx) $ &
 $ \paramx $ occurs at most once in the trace
\\

$\End{\paramx}$ &
 $ \ltlalws \ltlevtly \paramx $ & 
 The last event of any trace is $ \paramx $
\\

% $\Prec{\paramx}{\paramy}$ &
%% $ \lnot\paramy \ltlwntil \paramx $ &
%$ \ltlalws \left( \paramy \limply \ltlyday \ltlonce %\paramx \right)$ &
%$ \paramy $ &
%$ \ltlonce \paramx $
%\\

$\AltPrec{\paramy}{\paramx}$ &
% $ (\lnot\paramy \ltlwntil \paramx) \land \ltlalws(\paramy \limply \ltlnext  (\lnot\paramy \ltlwntil \paramx))$ &
 $ \ltlalws \left( \paramx \limply \ltlyday( \lnot\paramx \ltlsince \paramy ) \right) $ &
 Every occurrence of $ \paramx $ requires that $ \paramy $ occurred before, with no recurrence of $ \paramx $ in between
\\

\bottomrule
\end{tabular}
		\end{scriptsize}
	\end{adjustbox}
\end{table}
%
%\Cref{tab:declare:semantics} shows t
The semantics of a {\Declare} template is given as an {\LTLf} formula. % using two different forms of expression.
Given the free variables (``\emph{parameters}'') $\paramx$ and $\paramy$, e.g., $\AltPrec{\paramy}{\paramx}$ corresponds to $\ltlalws( \paramx \limply \ltlyday (\neg \paramx \ltlsince \paramy) )$, witnessing that for every instant in which $\paramx$ is verified, then a previous instant must verify $\paramy$ without any occurrences of $\paramx$ in between. Hitherto, we will occasionally use an abbreviation for the template name --- {\AltPrecShort{\paramy}{\paramx}}.
\Cref{tab:declare:semantics} shows the {\LTLf} formulae of some templates of the {\Declare} repertoire.
Standard {\Declare} imposes that template parameters be interpreted as single symbols of $\AuAlph$ to build \emph{constraints}.
For example, $\AltPrecShort{\taskb}{\taskc}$ interprets $\paramx$ as $\letterc$ and $\paramy$ as $\letterb$.
Branched {\Declare}~\cite{Pesic/2008:ConstraintbasedWorkflow} comprises the same set of templates of standard {\Declare}, yet allowing the interpretation of parameters as elements of a join-semilattice $\left( \LogAlph, \vee \right)$, i.e., an idempotent commutative semigroup, where $\vee$ is the join-operation~\cite{DiCiccio.etal/IS2016:EfficientDiscoveryTarget}.
%Birkhoff/ColloquiumPublications1967:LatticeTheory
We shall use a clausal set-notation whenever a parameter is interpreted as a disjunction of literals. 
For example, $\AltPrecShort{\{\lettera,\letterw\}}{\letterb}$ interprets $\paramx$ as $\letterb$ and $\paramy$ as $\lettera \vee \letterw$: for every $\letterb$ occurring in a trace, a previous instant must have verified $\lettera \vee \letterw$, without $\letterb$ recurring between that instant and the following occurrence of $\letterb$.
The conjunction of a finite set of constraints forms a {\Declare} specification.
In the following, we formalize the above notions.

\begin{definition}
    \label{def:decspec}
    A {\Declare} specification $\DecSpec = (\DecRepertoire,\AuAlph,\CnsSet)$ is a tuple wherein: \quad
    \begin{compactdesc}
        \item[\textnormal{\textit{$\DecRepertoire$}}] is a finite non-empty set of \emph{templates}, or ``\emph{repertoire}'', where each template $\DecTemp(x_1, \ldots, x_m)$ is an {\LTLf} formula parameterized on free variables $x_1, \ldots, x_m$; %, with $m$ being the \emph{arity} of $\DecTemp$;
        \item[$\AuAlph \ni \lettera_i $] is a finite non-empty alphabet of \emph{symbols} $\lettera_i$ with $1 \leq i \leq |\Sigma|$, $|\Sigma| \in \mathbb{N}$; 
        \item[$\CnsSet$] is a finite set of \emph{constraints}, namely pairs $(\DecTemp(x_1, \ldots, x_m), \kappa)$ where $\DecTemp(x_1, \ldots, x_m)$ is a template from $\DecRepertoire$, and $\kappa : \{ x_1, \ldots, x_m \} \to 2^{\AuAlph} \setminus \{\}$ is a mapping from every variable $x_i$ to a non-empty, finite set of symbols $\letterA_i = \{ \lettera_{i,1}, \ldots, \lettera_{i,v_i} \} \subseteq \AuAlph$, with $1 \leq i \leq m$ and $1 \leq v_i \leq |\Sigma|$; we denote such a constraint with $\Cns( \letterA_1, \ldots, \letterA_m )$ or equivalently $\Cns(\{ \lettera_{1,1}, \ldots, \lettera_{1,v_1} \}, \ldots, \{ \lettera_{m,1}, \ldots, \lettera_{m,v_m} \})$, omitting curly brackets from the latter form whenever a variable is mapped to a singleton.
    \end{compactdesc}
\end{definition}
%\begin{example}
\begin{sloppypar}	
\noindent
As an example, consider the following specification: $\DecRepertoire = \{ \AtMo{\paramx},$ $\End{\paramx},$ $\AltPrecShort{\paramy}{\paramx} \}$, $ \AuAlph = \{ \lettera,$ $\letterb,$ $\letterc,$ $\letterd,$ $\lettere,$ $\letterf,$ $\letterg,$ $\letteru,$  $\letterv,$ $\letterw \} $, and 
$ \CnsSet = \{ \AtMo{\lettera} $,
$\End{\letterv}$,
$\AltPrecShort{\lettere}{\letterf}$,
$\AltPrecShort{\{\lettera, \letterw\}}{\letterb} $,
$\AltPrecShort{\letteru}{\{\letterv, \letterw\}} \} $, where, e.g., $\AltPrecShort{\{\lettera, \letterw\}}{\letterb} \} $ is derived from the template $\AltPrecShort{\paramy}{\paramx}$ by mapping $\paramy \mapsto_\kappa \{\lettera, \letterw\}$ and $\paramx \mapsto_\kappa {\letterb}$.
\end{sloppypar}
%\end{example}

\begin{sloppypar}
\begin{definition}[Constraint formula, satisfying trace]
\label{def:constraint-formula}
Let $\Cns(\letterA_1, \ldots, \letterA_m)$ be a constraint, whereby
$\letterA_i = \{\lettera_{i,1}, \ldots, \lettera_{i,v_i}\}$
for each $1 \leq i \leq m$. Its \emph{constraint formula}, written $\FormulaOf{\Cns(\letterA_1, \ldots, \letterA_m)}$, is the {\LTLf} formula obtained from the template $\DecTemp(\paramx_1, \ldots, \paramx_m)$ by interpreting $\paramx_i$ as $\left(\lettera_{i,1} \vee \cdots \vee \lettera_{i,v_i}\right)$ for each $1 \leq i \leq m$. A trace $\ltrace$ \emph{satisfies} $\Cns(\letterA_1, \ldots, \letterA_m)$ iff $\ltrace \models \FormulaOf{\Cns(\letterA_1, \ldots, \letterA_m)}$; otherwise, we say that $\ltrace$ \emph{violates} $\Cns(\letterA_1, \ldots, \letterA_m)$.
\end{definition}
%
%\begin{example}
\noindent
Considering \cref{tab:declare:semantics} and the above example specification, we have that $\FormulaOf{\AltPrecShort{\{\lettera,\letterw\}}{\letterb}} = \ltlalws( \letterb \limply \ltlyday (\neg \letterb \ltlsince (\lettera \vee \letterw)) )$, and $\FormulaOf{\End{\letterv}} = \ltlalws \ltlevtly \letterv$.
Traces
$\langle \lettera, \letterb, \letterc \rangle$,
$\langle \lettera,\letterb,\letterc,\letterf,\letteru,\letterw,\letterb \rangle$, and 
$\langle \lettera,\letterb,\letterc,\lettere,\letterf,\letterg,\letteru,\letterv \rangle$ satisfy $\AltPrecShort{\{\lettera,\letterw\}}{\letterb}$,
while only the third one satisfies
$\End{\letterv}$.

\end{sloppypar}
%\end{example}

\begin{sloppypar}
   \begin{definition}[Specification formula, model trace]\label{def:modeltrace}
	A given {\Declare} specification $\DecSpec = (\DecRepertoire,\AuAlph,\CnsSet)$ is logically represented by conjoining its constraint formulae 
    ${\varphi_{\DecSpec} \doteq \bigwedge_{\Cns(\letterA_1, \ldots, \letterA_m) \in \CnsSet}\left( 
	\FormulaOf{\Cns(\letterA_1, \ldots, \letterA_m)} \right)}$. A trace is a \emph{model trace} for the specification, $\ltrace \models \DecSpec$, iff $\ltrace \models \varphi_{\DecSpec}$, i.e., it satisfies the conjunction of all the constraint formulae, $\ltrace \models \FormulaOf{\Cns(\letterA_1, \ldots, \letterA_m)}$ for each $\Cns(\letterA_1, \ldots, \letterA_m) \in \CnsSet$.
    \end{definition} 
\end{sloppypar}

\noindent
The specification formula of the above example is \begin{align*}
&\left(\ltlalws (\lettera \limply \lnot \ltlnext\ltlevtly\lettera)\right) \land\\
&\left(\ltlalws \ltlevtly \letterv \right) \land \\
&\left( \ltlalws \left( \letterf \limply \ltlyday( \lnot\letterf \ltlsince \lettere ) \right) \right) \land 
\left( \ltlalws \left( \letterb \limply \ltlyday( \lnot\letterb \ltlsince \left( \lettera \lor \letterw \right) ) \right) \right) \land\\
&\left( \ltlalws \left( \left( \letterv \lor \letterw \right) \limply \ltlyday( \lnot\left( \letterv \lor \letterw  \right) \ltlsince \letteru ) \right) \right)\end{align*}
$\langle \lettera,\letterb,\letterc,\lettere,\letterf,\letterg,\letteru,\letterv \rangle$ is a model trace for it, unlike $\langle \lettera, \letterb, \letterc \rangle$
or
$\langle \lettera,\letterb,\letterc,\letterf,\letteru,\letterw,\letterb \rangle$.

Leveraging the techniques mentioned at the end of \cref{sec:fsa} and the above definition, we can create an FSA that accepts all and only the traces of a single {\Declare} formula $\varphi$ and of a whole specification {\DecSpec}.

\begin{definition}[Constraint and specification FSA]\label{def:specification-fsa}
	Let $\varphi_1, \ldots, \varphi_{|\CnsSet|}$ be the constraint formulae of a process specification {\DecSpec}. A \emph{constraint automaton} $\Au_{\varphi_i}$ is an \gls{fsa} that accepts all and only those traces that satisfy $\varphi_i$~\cite{DiCiccio.etal/IS2017:ResolvingInconsistenciesRedundanciesDeclare} with $1\leq i \leq |\CnsSet|$. %$\Au_\varphi$ can be obtained through different algorithmic techniques.
	The product automaton $\Au_{\varphi_1} \times \cdots \times \Au_{\varphi_{|\CnsSet|}}$ is the \emph{specification FSA}, recognizing all and only the traces satisfying {\DecSpec}.
\end{definition}

\noindent
\Cref{fig:constraint:automata:atmoone:inst,fig:constraint:automata:altprec:inst,fig:constraint:automata:end:inst} show the automata of constraints \AtMo{\lettera}, $\AltPrecShort{\letteru}{\{\letterv,\letterw\}}$, and \End{\letterv}, respectively. \Cref{fig:constraint:automata:product:inst} depicts the FSA of a specification consisting of \End{\letterv} and $\AltPrecShort{\letteru}{\{\letterv,\letterw\}}$. Notice that the accepting state cannot be reached from $\AuSt_2$ in \cref{fig:constraint:automata:atmoone:inst,fig:constraint:automata:altprec:inst}. Instead, the FSA in \cref{fig:constraint:automata:product:inst} has no such trap states due to trimming.

Aside from keeping the FSA's language unchanged, trimming caters for structural compatibility with the state space representation of Workflow nets, which we discuss next.

\subsection{Workflow nets}
\label{sec:wf-net}
\begin{figure}[tb]
	\centering
	\resizebox{\textwidth}{!}{%
		\tikzset{
    -|/.style={to path={-| (\tikztotarget)}},
    |-/.style={to path={|- (\tikztotarget)}},
}
\begin{tikzpicture}[>=stealth',x=1cm,y=0.5cm,bend angle=90]
\begin{scriptsize}
%  \tikzstyle{place}=[circle,thick,minimum size=6mm]
\tikzstyle{transition}=[rectangle,thick,draw=black!75,minimum height=2em,minimum width=2em,text width=1em,align=center]
\tikzstyle{silenttransition}=[rectangle,thick,fill=black,minimum height=2em]
\tikzstyle{outputplace}=[double,draw,circle,minimum height=2em]

\node [place,tokens=1] (0) at (  0,  0)  [label={below:$p_0$},label={above:$\InPl$}] {}; %,tokens=1] {$\InPl$};
\node [place] (1) at ( +2, 0)   [label={$p_1$}]          {};
\node [place] (2) at ( +4, 0)   [label={below:$p_2$}]          {};
\node [place] (3) at ( +6, 0)   [label={below:$p_3$}]          {};
\node [place] (4) at ( +8, +1)  [label={below:$p_4$}]          {};
\node [place] (5) at ( +8, -1)  [label={below:$p_5$}]          {};
\node [place] (6) at ( +10, +1) [label={below:$p_6$}]          {};
\node [place] (7) at ( +10, -1) [label={below:$p_7$}]          {};
\node [place] (8) at ( +12, 0)  [label={$p_8$}]           {};
\node [outputplace] (9) at ( +14, 0) [label={below:$p_9$},label={above:$\OutPl$}] {};

\node [transition] (ta) at ( +1,  0) {$\Tx{\taska}$}
edge [pre]                 (0)
edge [post]                (1);
\node [transition] (tb) at ( +3, 0) {$\Tx{\taskb}$}
edge [pre]                 (1)
edge [post]                (2);
\node [transition] (tc) at ( +5, +1) {$\Tx{\taskc}$}
edge [pre]                 (2)
edge [post]                (3);
\node [transition] (td) at ( +5, -1) {$\Tx{\taskd}$}
edge [pre]                 (2)
edge [post]                (3);
\node [transition] (te) at ( +7, 0) {$\Tx{\taske}$}
edge [pre]                 (3)
edge [post]                (4)
edge [post]                (5);
\node [transition] (tf) at ( +9, +1) {$\Tx{\taskf}$}
edge [pre]                 (4)
edge [post]                (6);
\node [transition] (tg) at ( +9, -1) {$\Tx{\taskg}$}
edge [pre]                 (5)
edge [post]                (7);
\node [transition] (tu) at (+11,  0) {$\Tx{\tasku}$}
edge [pre]                 (6)
edge [pre]                 (7)
edge [post]                (8);
\node [transition] (tv) at ( +13, 0) {$\Tx{\taskv}$}
edge [pre]                 (8)
edge [post]                (9);
\node [transition] (tr) at ( +11, -2.5) {$\Tx{\taskw}$}
edge [pre,-|]                 (8)
edge [post,-|]                (1);

\end{scriptsize}
\end{tikzpicture}
	}
	\vspace{-4ex}\caption[A Workflow net]{A Workflow net}
	\label{fig:wnexample}
\end{figure}
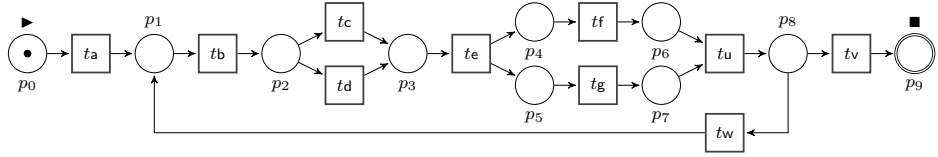
A Workflow net (see, e.g., \cref{fig:wnexample}) is a renowned subclass of Petri nets suitable for the formal representation of imperative process models~\cite{DBLP:journals/jcsc/Aalst98}. 

\begin{definition}[Petri net]\label{def:petrinet}
    A place/transition net~\cite{DeselReisig/ACPN1998:PlaceTransitionPetriNets} (henceforth, \emph{Petri net}) is a bipartite graph $ ( \Places, \Transitions, \FlowRel) $, where $\Places$ (the finite set of ``\emph{places}'') and $\Transitions$ (the finite set of ``\emph{transitions}'') constitute the nodes ($\Places \cap \Transitions = \emptyset$),
    and the \emph{flow relation} $\FlowRel \subseteq ( \Places \times \Transitions \uplus \Transitions \times \Places )$ defines the edges. %, i.e., a nonempty set of directed edges connecting places to transitions and transitions to places. 
    % Let {\AuAlph} be an alphabet, and $\SilentAct \notin \AuAlph$ a \emph{silent} action.
    % A \emph{Petri net} is a labeled P/T net $\Pn = \left( \Places, \Transitions, \FlowRel, \PnAlph, \LabelFunc \right)$ where $\PnAlph = \AuAlph \cup \{\SilentAct\}$ is the set of labels of a Petri net, and $\LabelFunc : \Transitions \to \PnAlph$ is the \emph{labeling function} of transitions.
\end{definition}

\noindent 
Given a place $ \Place \in \Places $, we shall denote the sets $ \{ \Transition \mid (\Transition, \Place) \in \FlowRel \} $ and $ \{ \Transition \mid (\Place, \Transition) \in \FlowRel \}$ with $ \Pre{\Place} $  (``\emph{preset}'') and $ \Post\Place $ (``\emph{postset}''),  respectively. For example, in \cref{fig:wnexample}, $\Pre{\Place_3}=\{\Tx\taskc\}$ and $\Post{\Place_3}=\{\Tx\taske,\Tx\taskf\}$. 
% We say that a place $\Place$ \emph{precedes} a transition $\Transition$ if $\Transition \in \Post\Place$ and, analogously, that $\Place$ \emph{succeeds} $\Transition$ if $\Transition \in \Pre\Place $. \todo{Never used in the text.}

% A Workflow net is a Petri net that satisfies three specific conditions.
\begin{definition}[Workflow net]\label{def:workflownet}
	A Workflow net $\Wn = ( \Places, \Transitions, \FlowRel )$ is a Petri net such that:
	\begin{compactenum}
		\item There is a unique place (``\emph{initial place}'', $ \InPl \in \Places$) such that its preset is empty;
		\item There is a unique place (``\emph{output place}'', $ \OutPl \in \Places$) such that its postset is empty;
		\item Every place $\Place \in \Places$ and transition $\Transition \in \Transitions$ is on a path of the underlying graph from $\InPl$ to $\OutPl$. 
	\end{compactenum}
\end{definition}
\noindent
We remark that we operate with full knowledge of the imperative model's structure, treated as a white box. Therefore, we directly focus on transitions rather than on their labels here.

In Petri and Workflow nets, places can be \emph{marked} with tokens, intuitively representing resources that are processed by the transitions succeeding them in the net. In \cref{fig:wnexample}, a token is graphically depicted as a solid circle (see $\Place_0$ in the figure).
The state of a net is defined by the distribution of tokens over places. This is formalized with the notion of \emph{marking}, a function mapping each place to the number of tokens in it. The net's state changes with the consumption and production of tokens caused by the execution (``\emph{firing}'') of transitions.
%
\begin{comment}
	Tokens are generated when transitions in the Petri net are ``\emph{fired}''.
	A transition $ \Transition $ can be fired if and only if all the places $ \Place $ that precede it have at least one token. Formally, given a marking $ \Marking $ of a Petri net \Pn, a transition $ \Transition $ can be fired if and only if $ \Marking(\Place) > 0$ for all places $ \Place $ such that $ \Transition \in \Post\Place $; in this case we say that the marking $ \Marking $ \emph{enables} transition $ \Transition $. When a transition is fired, the marking is modified by removing one token from all places preceding the transition and adding one token to all the places succeeding it.
\end{comment}
%
\begin{comment}
	For Workflow nets, we will only consider firing sequences that end in a so-called ``\emph{final} marking'' $\Marking_o$,
	i.e., a marking in which the output node has a number of tokens greater than 0: $\Marking_o(\OutPl) > 0$. \todo{In the literature, firing sequences of Workflow nets do not necessary end in the final marking, alas. We cannot overwrite this notion.}
\end{comment}

\begin{definition}[Marking and firing]
	Let $\Wn = ( \Places, \Transitions, \FlowRel )$ be a Workflow net.
	A \emph{marking} is a function $\Marking : \Places \to \mathbb{N} \cup \{0\}$.
	The \emph{initial marking} $\Marking_0$ of $\Wn$ maps $\InPl$ to $1$ and any other $\Place \in \Places \setminus \{\InPl\}$ to $0$.
	A marking $\Marking$ of $\Wn$ is \emph{final} if $\Marking(\OutPl) > 0$.
	A marking \emph{enables} a transition $\Transition \in \Transitions$ iff $ \Marking(\Place) > 0$ for all places $ \Place $ such that $ \Transition \in \Post\Place $.
	An enabled transition can \emph{fire}, i.e., turn a marking $\Marking$ into $\Marking'$ (in symbols, $\Fires{\Marking}{\Transition}{\Marking'}$), according to the following rule: For each place $\Place \in \Places$,
	$\Marking'(\Place) = \Marking(\Place) + 1$ if $\Transition \in \Pre\Place$;
	$\Marking'(\Place) = \Marking(\Place) - 1$ if $\Transition \in \Pre\Place$;
	otherwise, $\Marking'(\Place) = \Marking(\Place)$.
\end{definition}   

\noindent
In \cref{fig:wnexample}, e.g., the initial marking enables $\Tx\taska$. Denoting markings with a multi-set notation,
$
\Fires{\left\{ \Place_0 \right\}}{\Tx\taska}{\left\{ \Place_1 \right\}}
$.
Subsequently, $\Tx\taskb$ gets enabled. After firing $\Tx\taskb$, and $\Tx\taskc$ get enabled. With Petri and Workflow nets, interleaving semantics are adopted, thus only one transition can fire per timestep, thus the firing of $\Tx\taskb$ and $\Tx\taskc$ are mutually exclusive in that state.
\begin{definition}[Firing sequence and run]
    Given a Workflow net $\Wn = ( \Places, \Transitions, \FlowRel )$, a
    (finite) \emph{firing sequence} 
    $ \PnFireSeq $ is $\langle \rangle$ or a sequence of transitions $ \langle \Transition_1, \ldots, \Transition_n \rangle $ such that, for any index $ 1 \leq i \leq n $ with $ n \in \mathbb{N}$, $\Transition_i \in \Transitions$:
    \begin{inparaenum}[(a)]
       \item The $i$-th marking enables the $i$-th transition;
       \item The $i + 1$-th marking $\Marking_{i + 1}$ is such that $\Fires{\Marking_i}{\Transition_{i}}{\Marking_{i + 1}} $.
    \end{inparaenum}
    A marking $\Marking'$ is \emph{reachable} in $\Wn$ if there exists a firing sequence $ \PnFireSeq $ leading from the initial marking $\Marking_0$ to $\Marking'$ (in symbols, $\Fires{\Marking_0}{\PnFireSeq}{\Marking'} $).
    A firing sequence leading from $\Marking_0$ to a final marking is a \emph{run}.
\end{definition}
\begin{comment} % Definition removed as it is not used later in the paper.
	We denote with $\PnSetFireSeq(\Pn)$ the set of all firing sequences of the Petri net \Pn.
\end{comment}

\noindent
Given a workflow net $\Wn$, we will use $\ReachMarkings{\Wn}$ to denote the set of markings that can be reached from its initial marking $\Marking_0$.

Runs of the Workflow net in \cref{fig:wnexample} include
$\langle \Tx\taska, \Tx\taskb, \Tx\taskc, \Tx\taske, \Tx\taskf, \Tx\taskg, \Tx\tasku, \Tx\taskv \rangle$
and
$\langle \Tx\taska, \Tx\taskb, \Tx\taskd, \Tx\taske, \Tx\taskf, \Tx\taskg, $ 
$\Tx\tasku, \Tx\taskw, \Tx\taskc, \Tx\taske, \Tx\taskf, \Tx\taskg, \Tx\tasku, \Tx\taskv \rangle$.
Any prefix of the first run of length \num{7} or less is a firing sequence from the initial marking $\left\{ \Place_0 \right\}$ but not a run.

In this paper, we assume that Workflow nets enjoy the following properties.
\begin{definition}[Soundness and safety of a Workflow net]
    Let $\Wn = ( \Places, \Transitions, \FlowRel )$ be a Workflow net. 
    \Wn is \emph{$k$-bounded} if the number of tokens assigned by any reachable marking $\Marking'$ to any place $p \in P$ is such that $\Marking'(p) \leq k$. A 1-bounded Workflow net is \emph{safe}.
    \Wn is \emph{sound} iff it enjoys the following properties:
    \begin{inparadesc}
        \item[Option to complete:] from any marking $\Marking$ it is possible to reach the final marking;
        \item[Proper completion:] if a reachable marking $\Marking$ is such that $\Marking(\OutPl) > 0$, then $\Marking$ is the final marking;
        \item[No dead transitions:] for any transition $\Transition \in \Transitions$, there exists a reachable marking $\Marking$ such that $\Transition$ is enabled by $\Marking$.
    \end{inparadesc}

\end{definition}

Safe and sound  Workflow nets (like the one depicted in \cref{fig:wnexample}) are a superclass of sound S-coverable nets, which in turn subsume safe and sound free-choice and well-structured nets~\cite{Aalst/1996:StructuralCharacterizationsSoundWfNs}.
These structural characteristics are widely recognized as recommendable in process management~\cite{DBLP:journals/jcsc/Aalst98} and underpin well-formed business process diagrams~\cite{Ouyang.etal/TOSEM2009:BusinessProcessModels2ProcessOrientedSwSystems}.
% \todo{To be extended with a few more details about S-coverability. Also, a note quoted from \cite{Weidlich.etal/BPM2010:DecidingBehaviourCompatibilityProcessModels}: ``sound free-choice WF-nets show a tight coupling of syntax and semantics. In particular, if N is sound and free-choice, the existence of a path $\pi(x, y)$ between places $x$ and $y$ implies the existence of a firing sequence containing all transitions on $\pi(x, y)$''}
Notice that, given a \emph{safe} net, all markings that are reachable from the initial one are such that each place can be marked with at most one token. Also, the final marking of a sound workflow net is $\{\OutPl\}$.
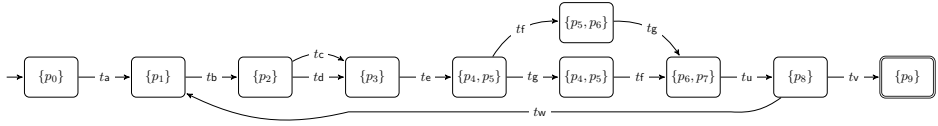
\begin{figure}[tb]
	\centering
	\resizebox{\textwidth}{!}{
		\begin{tikzpicture}[->,>=stealth',shorten >=1pt,node distance=8em,
                    thick,
                    initial text = {},
                    state/.style = {
                        rounded corners,
                        minimum width=4em,
                        minimum height=3em,
                        draw,
                        inner sep=1mm,
                    },]
      \node[state,initial] (0) { $\left\{\Place_0\right\}$};
      \node[state] (1) [right of=0] {$\left\{\Place_1\right\}$};
      \node[state] (2) [right of=1] {$\left\{\Place_2\right\}$};
      \node[state] (3) [right of=2] {$\left\{\Place_3\right\}$};
      \node[state] (45) [right of=3] {$\left\{\Place_4, \Place_5\right\}$};
      \node[state] (47) [right of=45] {$\left\{\Place_4, \Place_5\right\}$};
      \node[state] (56) [above =1em of 47] {$\left\{\Place_5, \Place_6\right\}$};
      \node[state] (67) [right of=47] {$\left\{\Place_6, \Place_7\right\}$};
      \node[state] (8) [right of=67] {$\left\{\Place_8\right\}$};
      \node[accepting, state] (9) [right of=8] {$\left\{\Place_{9}\right\}$};
      
      \path
         (0) edge node [midway, fill=white] {$\Tx{\taska}$} (1)
         (1) edge node [midway, fill=white] {$\Tx{\taskb}$} (2)
         (2) edge node [midway, fill=white] {$\Tx{\taskd}$} (3)
         (2) edge [bend left] node [midway, fill=white] {$\Tx{\taskc}$} (3)
         (3) edge node [midway, fill=white] {$\Tx{\taske}$} (45)
        
        (45) edge node [midway, fill=white] {$\Tx{\taskg}$} (47)
        (45) edge [bend left] node [midway, fill=white] {$\Tx{\taskf}$} (56)
        (56) edge [bend left] node [midway, fill=white] {$\Tx{\taskg}$} (67)
        (47) edge node [midway, fill=white] {$\Tx{\taskf}$} (67)        
        
        (67) edge node [midway, fill=white] {$\Tx{\tasku}$} (8)
         
%         (8) edge [out=-160,in=-20] node [midway, fill=white] {$\Tx{\taskw}$} (1)
                
         (8) edge node [midway, fill=white] {$\Tx{\taskv}$} (9);
         
              	\draw[->]
         (8) to [bend left=20] ++ (-2,-1) to node[midway, fill=white] {$\Tx{\taskw}$} ++ (-11, 0) to [bend left=20] (1);

%
%        (3) edge [out=90,in=90,looseness=0.9] node [midway, fill=white] {$\Tx{\taske}$} (4)
%        (3) edge [out=-90,in=-90,looseness=0.9] node [midway, fill=white] {$\Tx{\taskf}$} (4)
%
%        (4) edge node [midway, fill=white] {$\Tx{\taski}$} (5)
%        
%        (5) edge node [midway, fill=white] {$\Tx{\tau_1}$} (6)
%
%        (6) edge node [midway, fill=white] {$\Tx{\taskm}$} (8)
%        (6) edge [out=90,in=180] node [midway, fill=white] {$\Tx{\taskk}$} (7)
%
%        (7) edge [out=0,in=90] node [midway, fill=white] {$\Tx{\taskm}$} (9)
%
%        (8) edge node [midway, fill=white] {$\Tx{\taskk}$} (9)
%
%        (9) edge node [midway, fill=white] {$\Tx{\tau_2}$} (10)
%
%        (10) edge node [midway, fill=white] {$\Tx{\taskj}$} (11)
%        
%        (11) edge node [midway, fill=white] {$\Tx{\taskb}$} (12)
%        (11) edge [out=-90,in=-90,looseness=0.3] node [midway, fill=white] {$\Tx{\taskg}$} (2);
\end{tikzpicture}
	}
	\vspace{-4ex}
	\caption[Reachability graph]{Reachability graph derived from the Workflow net in \cref{fig:wnexample}.} % Each node represents a marking, where the label in the node is a set specifying which places have one token assigned to it in the corresponding marking. Each edge between two markings is labeled with the name of the transition that modifies the first marking into the second}
	\label{fig:reachabilitygraph}
\end{figure}

The state space of $k$-bounded Petri nets can be represented in the form of a deterministic labeled transition system that go under the name of \emph{reachability graph}~\cite{DeselReisig/ACPN1998:PlaceTransitionPetriNets}. Safe and sound Workflow nets have a given initial marking and one final marking. We can thus endow the state representation with these characteristics and reinterpret the known concept of reachability graph as a finite state automaton.
% \todo{CDC4Giovanni: Could you please merge the definition of reachability graph and following FSA-transformation? Notice that anyway the reinterpretation would have not been possible for any Pn or even any WfN, because for the former there is no final marking defined, and for the latter if the final marking is unreachable an accepting state in the reach.graph. does not exist.}
\begin{definition}[Reachability FSA]\label{def:reachability-graph}
    Given a sound and safe Workflow net $\Wn$, the reachability FSA
    $\Au^{\Wn} = \left( \AuStSet^{\Wn}, \AuStInit^{\Wn}, \AuStAcc^{\Wn}, \AuAlph^{\Wn}, \AuTrns^{\Wn} \right) $ is a finite state automaton where:
    \begin{compactdesc}
        \item[$\AuStSet^{\Wn} = \ReachMarkings{\Wn}$], i.e., the set of states is the set of reachable markings in $\Wn$;
        \item[$\AuStInit^{\Wn} = \{\InPl\}$], i.e., the initial state is the initial marking of $\Wn$;
        \item[$\AuStAcc^{\Wn} = \{\{\OutPl\}\}$], i.e., the accepting state set is a singleton containing the final marking of $\Wn$; % $\AuStAcc = \{ \Marking \mid \neg\exists \Marking' \in \ReachMarkings{\Wn} \exists \Transition \in \Transitions. \Fires{\Marking}{\Transition}{\Marking'}\}$, i.e., $\AuStAcc$ contains all markings which enable no transition. Note that, for the case of S-coverable workflow nets, we have that $\AuStAcc = \{ \Marking_o \}$, i.e., the only accepting state is the one corresponding to the final marking of the net\todo{GV: doesn't this also apply to sound and safe workflow nets? The only marking that doesn't allow any transition (at least to me) seems to be the one where only the output place has a token, i.e., the final marking of the workflow net.}\todo{CDC: Exactly. Then, removed.};uch that
        \item[$\AuAlph^{\Wn} = \Transitions$], i.e., the alphabet is the set of transitions of $\Wn$;
        \item[$\AuTrns^{\Wn}$] is s.t.\ $\AuTrns(\Marking, \Transition) = \Marking'$ iff $\Fires{\Marking}{\Transition}{\Marking'}$ for every transition $t$ and reachable $\Marking,\Marking'$ in \Wn.
    \end{compactdesc}
\end{definition}
%
% When clear from the context, we will drop the superscript for the given marking. This is especially the case for workflow nets, given that each workflow net has a unique initial marking $\Marking_i$. Note that in workflow nets we will denote with $\ReachMarkings{\Wn}$ the set of markings that are reachable from the initial marking $\Marking_i$.
%
% It is worthwhile to note that a reachability graph stemming from a Workflow net can be legitimately augmented to a deterministic \gls{fsa}. This observation follows from the fact that each transition alters markings in a deterministic manner;
% \todo{Which one?}
% Given an arbitrary Petri net $\Pn$ and a marking $\Marking_0$ for it, we can define the automaton $\Au^{\Pn}_{\Marking_0} = \left( \AuStSet, \AuStInit, \AuStAcc, \AuAlph, \AuTrns \right) $ where:
% $\AuStSet := \ReachMarkings{\Pn}[\Marking_0] $, $\AuStInit := \Marking_0$, , $\AuAlph := \Transitions$, and $\AuTrns(\Marking, \Transition) := \Marking'$, where $\Marking$ enables $\Transition$ and $\Fires{\Marking}{\Transition}{\Marking'}$.

\noindent \Cref{fig:reachabilitygraph} depicts the reachability FSA of the Workflow net in \cref{fig:wnexample}. 

When dealing with accepting traces and languages, working with trimmed or non-trimmed FSAs is equivalent. This is not the case when we structurally relate the FSA of a {\Declare} specification with the reachability FSA of a Workflow net. That FSA is indeed constructed, state-by-state, considering only enabled transitions, which globally yields that it is trimmed by design. Therefore we operate with trimmed FSAs for this comparison.

\begin{comment} \todo{The statement below does not belong here. It already mentions bisimilarity among different behavioural representations, which is undefined in this context.}
Our main theoretical result will consist in proving how, given a safe and sound Workflow net, we can produce a {\Declare} specification that is bisimilar to it. Typically, bisimilarity is a condition between dynamic systems which (intuitively) entails that to an external observer the two systems behave in an indistinguishable manner. In our case, bisimilarity signifies that the {\Declare} specification is satisfied by all and only the runs of the net. 
\end{comment}

% G: non penso ci serva questa roba, ma la lascio qua in caso
%It is clear that from a reachability graph $\ReachGraph{\Pn}[\Marking_0]$ one can obtain a labelled transition system (LTS) $L_{\Pn}^{\Marking_0} = (S, s_0, A, \Rightarrow)$ by setting $S := \ReachMarkings{\Pn}[\Marking_0]$, $s_0 := \Marking_0$, $A := \Transitions$, and $\Rightarrow := \{ (\Marking, \Transition, \Marking') \mid \Fires{\Marking}{\Transition}{\Marking'} \}$. Moreover, note that such LTS is deterministic, as any transition of any Petri net modifies any marking in a deterministic manner. This will prove crucial when we will establish how to derive a {\Declare} specification that is \emph{bisimilar} to an input sound workflow net. 

\section{Synthesis of LTL$_f$ specifications from Workflow nets}
\label{sec:algo}
%In this section, we outline how to derive {\Declare} specifications that are \emph{bisimilar} to input sound workflow nets {\Wn}s. \Cref{alg:imperative-to-declarative} illustrates how we generate a {\Declare} specification from an input workflow net \Wn.

\begin{algorithm2e}[tb]
	\caption{\small Wizard’s guide to synthesize {\Declare} specifications from Workflow nets}
	\label{alg:imperative-to-declarative}
	% !TeX program = pdflatex
% !BIB program = bibtex
% !TeX spellcheck = en_UK
% !TeX encoding = utf8
% !TeX root = ../main.tex
%
\scriptsize
\SetKwComment{Comment}{\#~}{}
\DontPrintSemicolon \SetAlgoVlined 
\KwIn{%
	$\Wn = \left( \Places, \Transitions, \FlowRel \right)$, a workflow net;%
	%  \newline%
}%
\KwOut{%
	$\DecSpec = \left( \DecRepertoire, \ProActvts, \CnsSet \right)$, a declarative process specification
}
\BlankLine
$\AuAlph \gets \Transitions; \quad \CnsSet \gets \{\}; \quad $ \Comment*{Assign the alphabet with the transition set, and initialize the constraint set}\label{alg:imp2dec:alph}
$\DecSpec \gets ( \{\AtMo{\paramx}, \End{\paramx}, \AltPrecShort{\paramx}{\paramy}\}, \, \AuAlph, \, \CnsSet )$ \Comment*{Initialize{\DecSpec} including templates}\label{alg:imp2dec:templates}
\ForEach(\Comment*[f]{Visit all places in \Wn}){$ \Place \in \Places$}{%
%	{$\AuAlph \gets \AuAlph \;\cup\; \Pre\Place \;\cup\; \Post\Place;$} \Comment*{Update the alphabet with transitions in the pre- and post-sets}\label{alg:imp2dec:alph}%
	\lIf(\Comment*[f]{Add the \AltPrecShort{\Pre\Place}{\Post\Place} constraint {\StarOne}}){$\Pre\Place \neq \emptyset$ and $\Post\Place \neq \emptyset$}{%
		$\CnsSet \gets \CnsSet \cup \left\{ \AltPrecShort{\Pre\Place}{\Post\Place} \right\}$}\label{alg:imp2dec:altprec}%
	\lElseIf(\Comment*[f]{Add the \AtMo{\Post\Place} constraint {\StarTwo}}){$\Pre\Place = \emptyset$}{%
		$\CnsSet \gets \CnsSet \cup \left\{ \AtMo{\Post\Place} \right\}\label{alg:imp2dec:atmo}$%
	}%
	\lElseIf(\Comment*[f]{Add the \End{\Pre\Place} constraint {\StarThree}}){$ \Post\Place = \emptyset$}{%
		$\CnsSet \gets \CnsSet \cup \left\{ \End{\Pre\Place} \right\} $\label{alg:imp2dec:end}%
	}% We assume there is no WfN without transitions (T is not empty) - I've just defined the elements of a Petri net to all be nonempty
}
\end{algorithm2e}
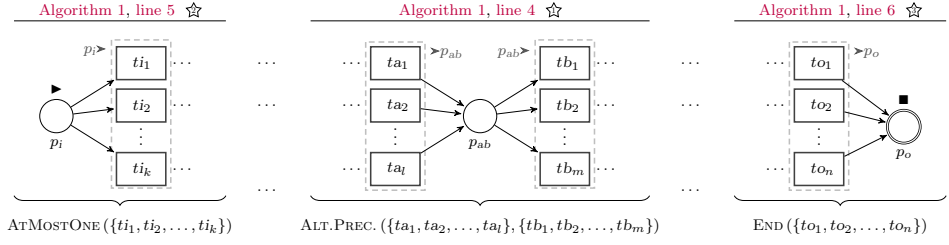
\begin{figure}[tb]
	\vspace{-2ex}
	\resizebox{\textwidth}{!}{%
		\begin{tikzpicture}[>=stealth',x=4em,y=2em]%,bend angle=45,auto,x=1.125cm,y=0.75cm,bend angle=90]
	\begin{scriptsize}
		%  \tikzstyle{place}=[circle,thick,minimum size=6mm]
		\tikzstyle{transition}=[rectangle,thick,draw=black!75,minimum height=2em,minimum width=3em,text width=1em,align=center]
		\tikzstyle{outputplace}=[double,draw,circle,minimum height=2em]

            \node [align=center] (atmo1step) at (+0.75, +3.25) {\Crefalgln[alg:imperative-to-declarative]{alg:imp2dec:atmo}~~\StarTwo};
            \draw [solid] (-0.5,+3) -- (2,+3);
        
		\node [place] (pi) at (0, +0.75)  [label={below:$\Place_\letteri$},label={above:$\InPl$}] {};

		\draw[color=gray!50, outer sep=1.8em, thick, densely dashed] (0.65,2.55) rectangle node[anchor=south east,yshift=1.4em,text=black!66]{$\Post{\Place_\letteri}$} ++(0.7,-3.5);
				
		\node [transition] (ti1) at ( +1,  +2) {$\Tx{\letteri_1}$}
		 edge [pre]                 (pi);
		\node [] (ti1dots) at (+1.5, +2) [] {$\cdots$};
		
		\node [transition] (ti2) at ( +1,  +1) {$\Tx{\letteri_2}$}
		 edge [pre]                 (pi);
		\node [] (ti2dots) at (+1.5, +1) [] {$\cdots$};
		
		\node [] (tidots) at ( +1,  0.4) {$\vdots$};
				
		\node [transition] (tim) at ( +1,  -0.5) {$\Tx{\letteri_k}$}
		 edge [pre]                 (pi);
		\node [] (timdots) at (+1.5, -0.5) [] {$\cdots$};
		
		\draw [decorate,decoration={brace,amplitude=5pt,mirror,raise=4ex}]
		(-0.5,-0.5) -- (2,-0.5) node[midway,yshift=-3.5em]{$\AtMo{ \{ \Tx{\letteri_1}, \Tx{\letteri_2}, \ldots, \Tx{\letteri_k} \} }$};

		\node [] (t1tdots) at (+2.5, +2) [] {$\cdots$};
		\node [] (t1tdots) at (+2.5, +1) [] {$\cdots$};
		\node [] (t1tdots) at (+2.5, -1) [] {$\cdots$};

            \node [align=center] (altprecstep) at (+5, +3.25) {\Crefalgln[alg:imperative-to-declarative]{alg:imp2dec:altprec}~~\StarOne};
            \draw [solid] (3,+3) -- (7,+3);

        \node [place] (p) at (+5, +0.75) [label={below:$\Place_{\lettera\letterb}$}] { };
		
		\draw[color=gray!50, outer sep=1.8em, thick, densely dashed] (3.65,2.5) rectangle node[anchor=south west,yshift=1.4em,text=black!66]{$\Pre{\Place_{\lettera\letterb}}$} ++(0.7,-3.5);
		
		\node [transition] (1t) at ( +4,  +2) {$\Tx{\lettera_1}$}
		edge [post]                 (p);
		\node [] (1tdots) at (+3.5, +2) [] {$\cdots$};
		
		\node [transition] (2t) at ( +4,  +1) {$\Tx{\lettera_2}$}
		edge [post]                 (p);
		\node [] (2tdots) at (+3.5, +1) [] {$\cdots$};
		
		\node [] (tdots) at ( +4,  0.4) {$\vdots$};
		
		\node [transition] (nt) at ( +4,  -0.5) {$\Tx{\lettera_l}$}
		edge [post]                 (p);
		\node [] (ntdots) at (+3.5, -0.5) [] {$\cdots$};
		
		\draw[color=gray!50, outer sep=1.8em, thick, densely dashed] (5.65,2.55) rectangle node[anchor=south east,yshift=1.3em,text=black!66]{$\Post{\Place_{\lettera\letterb}}$} ++(0.7,-3.5);

		\node [transition] (t1) at ( +6,  +2) {$\Tx{\letterb_1}$}
		edge [pre]                 (p);
		\node [] (t1dots) at (+6.5, +2) [] {$\cdots$};
		
		\node [transition] (t2) at ( +6,  +1) {$\Tx{\letterb_2}$}
		edge [pre]                 (p);
		\node [] (t2dots) at (+6.5, +1) [] {$\cdots$};
		
		\node [] (ttdots) at ( +6,  0.4) {$\vdots$};
		
		\node [transition] (tn) at ( +6,  -0.5) {$\Tx{\letterb_m}$}
		edge [pre]                 (p);
		\node [] (tndots) at (+6.5, -0.5) [] {$\cdots$};

		\draw [decorate,decoration={brace,amplitude=5pt,mirror,raise=4ex}]
		(3,-0.5) -- (7,-0.5) node[midway,yshift=-3.5em]{$\AltPrecShort{ \{ \Tx{\lettera_1}, \Tx{\lettera_2}, \ldots, \Tx{\lettera_l} \} }{ \{ \Tx{\letterb_1}, \Tx{\letterb_2}, \ldots, \Tx{\letterb_m} \} }$};

		\node [] (ttodots) at (+7.5, +2) [] {$\cdots$};
		\node [] (ttodots) at (+7.5, +1) [] {$\cdots$};
		\node [] (ttodots) at (+7.5, -1) [] {$\cdots$};

            \node [align=center] (altprecstep) at (+9.25, +3.25) {\Crefalgln[alg:imperative-to-declarative]{alg:imp2dec:end}~~\StarThree};
            \draw [solid] (8,+3) -- (10.5,+3);

		\node [outputplace] (po) at ( +10, 0.5)  [label={below:$\Place_\lettero$},label={above:$\OutPl$}] {};

		\draw[color=gray!50, outer sep=1.8em, thick, densely dashed] (8.65,2.5) rectangle node[anchor=south west,yshift=1.4em,text=black!66]{$\Pre{\Place_o}$} ++(0.7,-3.5);

		\node [transition] (ti1) at ( +9,  +2) {$\Tx{\lettero_1}$}
		edge [post]                 (po);
		\node [] (ti1dots) at (+8.5, +2) [] {$\cdots$};
		
		\node [transition] (ti2) at ( +9,  +1) {$\Tx{\lettero_2}$}
		edge [post]                 (po);
		\node [] (ti2dots) at (+8.5, +1) [] {$\cdots$};
		
		\node [] (tidots) at ( +9,  0.4) {$\vdots$};
		
		\node [transition] (tim) at ( +9,  -0.5) {$\Tx{\lettero_n}$}
		edge [post]                 (po);
		\node [] (timdots) at (+8.5, -0.5) [] {$\cdots$};
		
		\draw [decorate,decoration={brace,amplitude=5pt,mirror,raise=4ex}]
		(8,-0.5) -- (10.5,-0.5) node[midway,yshift=-3.5em]{$\End{ \{ \Tx{\lettero_1}, \Tx{\lettero_2}, \ldots, \Tx{\lettero_n} \} }$};
		
	\end{scriptsize}
\end{tikzpicture}%
	}
	\vspace{-4ex}\caption[A graphical sketch of the algorithm]{A graphical sketch of the execution of \cref{alg:imperative-to-declarative}} %
	\label{fig:imperative-to-declarative}
\end{figure}	

In this section, we outline the algorithm (including the three spells to cast: \StarOne, \StarTwo, and \StarThree) to synthesize a {\Declare} specification {\DecSpec} from a given input safe and sound Workflow net {\Wn}
% {\Wn}
%\todo{Here and everywhere: it should be sound \emph{and} safe. Just \emph{sound} does not suffice, alas.}
ensuring behavioral equivalence between them. \Cref{alg:imperative-to-declarative} illustrates the transformation process.
The algorithm initializes {\DecSpec} by assigning its alphabet with the transition set of {\Wn} (\crefalgln[alg:imperative-to-declarative]{alg:imp2dec:alph}). Given the Workflow net in~\cref{fig:wnexample}, e.g., $\AuAlph$ gets $\{\Tx\taska,\ldots,\Tx\taskg,\Tx\tasku,\Tx\taskv,\Tx\taskw\}$. Then, it sets the three (necessary) templates that will be used (\crefalgln{alg:imp2dec:templates}): $\AtMo{\paramx}, \End{\paramx},$ and $\AltPrecShort{\paramy}{\paramx}$. A cycle begins to visit all places in {\Wn} and update {\DecSpec} by including a new constraint per place. 
%For every place $\Place$, it first updates the alphabet of {\DecSpec} to include the transitions in the preset and postset of $\Place$ (see~\crefalgln{alg:imp2dec:alph}). This passage entails that $\AuAlph$ is ultimately assigned with $T=\{\}$.
% Then, it updates {\DecSpec} by including a new constraint.
\Cref{fig:imperative-to-declarative} graphically sketches this passage, which casts the three spells as follows:
\begin{inparaitem}
	\item[\StarThree] If $\Place$ is the output place as $\Post\Place$ is empty, $\End{\Pre\Place}$ is included (\crefalgln{alg:imp2dec:end});
	\item[\StarTwo] If $\Place$ is the input place as $\Pre\Place$ is empty, $\AtMo{\Post\Place}$ is added (\crefalgln{alg:imp2dec:atmo});
	\item[\StarOne] Otherwise, $\AltPrecShort{\Pre\Place}{\Post\Place}$ becomes one of the constraints in {\DecSpec} (\crefalgln{alg:imp2dec:altprec}).
\end{inparaitem}
Intuitively, the rationale is that:
\begin{inparaitem}
	\item[\StarOne] Every time a transition in the postset of $\Place$ fires, it is necessary that at least one of the transitions in the preset of $\Place$ fired before and that no transition in the postset of $\Place$ has fired since then;
	\item[\StarTwo] Any of the transitions in the postset of $\InPl$ will start the run and will not repeat afterwards (because no firing can assign $\InPl$ with a token again by definition); 
	\item[\StarThree] Every run must terminate with one of the transitions in the preset of $\OutPl$.
\end{inparaitem}

\begin{table}[tb]
    \centering
    \caption[Specification generated from the Workflow net]{{\Declare} specification generated from the Workflow net in \cref{fig:wnexample}\vspace{-2ex}}
    \label{tab:constraints}
    \resizebox{\columnwidth}{!}{%
\begin{tabular}{lllll}
    \toprule
    $\AtMo{\Tx{\taska}}$& 
    $\End{\Tx{\taskv}}$&
    $\AltPrecShort{\{\Tx{\taska},\Tx{\taskw}\}}{\Tx{\taskb}}$&
    $\AltPrecShort{\Tx{\taskb}}{\{\Tx{\taskd},\Tx{\taskc}\}}$& 
    $\AltPrecShort{\{\Tx{\taskd},\Tx{\taskc}\}}{\Tx{\taske}}$ 
    \\
    $\AltPrecShort{\Tx{\taske}}{\Tx{\taskf}}$&
    $\AltPrecShort{\Tx{\taske}}{\Tx{\taskg}}$& 
    $\AltPrecShort{\Tx{\taskf}}{\Tx{\tasku}}$& 
    $\AltPrecShort{\Tx{\taskg}}{\Tx{\tasku}}$&
    $\AltPrecShort{\Tx{\tasku}}{\{\Tx{\taskv},\Tx{\taskw}\}}$ \\
    \bottomrule
\end{tabular}
}
\end{table}
\Cref{tab:constraints} shows the constraints that are generated by our algorithm if the Workflow net in \cref{fig:wnexample} is fed in input. It is noteworthy to analyze in particular the non-trivial behavior entailed by constraints that stem from the parsing of places that begin or end cycles like $\Place_8$ and $\Place_1$ in \cref{fig:wnexample}. From the former we derive $\AltPrecShort{\Tx\tasku}{\{\Tx\taskv,\Tx\taskw\}}$. It states that before \Tx{\taskv} or \Tx{\taskw}, \Tx{\tasku} must occur. Also, \emph{neither} \Tx{\taskv} nor \Tx{\taskw} car recur until \Tx{\tasku} is repeated. As a consequence, an \emph{exclusive} choice between \Tx{\taskv} and \Tx{\taskw} is enforced cyclically for each recurrence of \Tx{\tasku}. Dually, with $\AltPrecShort{\{\Tx{\taska},\Tx{\taskw}\}}{\Tx{\taskb}}$ (generated by fetching the pre- and post-sets of $\Place_1$) we demand that \emph{each} occurrence of \Tx{\taskb} follows \Tx{\taska} or \Tx{\taskw}. From \cref{tab:constraints} we notice that \Tx{\taska} can occur only once (\AtMo{\Tx\taska}), thus the subsequent recurrences of \Tx{\taskb} are bound to \Tx{\taskw}.
Given the construction in \cref{alg:imperative-to-declarative}, it is clear that the semantics of the resulting {\Declare} specification $\DecSpec$ is an {\LTLf} formula, traces of which are finite sequences of transitions of the input Workflow net \Wn. 
Notice that the mapping of the transitions of \Wn to labels (as usual in a process modeling context) can be treated as a post-hoc refinement of \Wn and equivalently of {\DecSpec}: Assuming that $\Transition_1$ maps to label $\taskz$, e.g., the occurrence of transition $\Transition_1$ will emit $\taskz$ regardless of the underlying behavioral representation.

\medskip
As established in the beginning of this section, our goal is to now show the behavioral equivalence between the {\Declare} specification given in output by \cref{alg:imperative-to-declarative} and the input safe and sound Workflow net.
%\todo{CDC4All: Write `Worklow net' with the uppercase W. I tried to  turn it into a lowercase w but every now and then new W appeared. The mixture is not OK. Workflow net. Uppercase W.}
To this end, we use the following notion of bisimilarity.
\begin{definition}[\small{Bisimilarity of Workflow nets and {\Declare} specifications}]
    A safe and sound Workflow net $\Wn$ is \emph{bisimilar} to a {\Declare} specification $\DecSpec$ if and only if the reachability FSA of $\Wn$ (as per \cref{def:reachability-graph}) is bisimilar to the specification FSA of $\DecSpec$ (as per \cref{def:specification-fsa}). 
\end{definition}

Given
\begin{iiilist}
	\item this notion of bisimilarity, and 
	\item \cref{obs:deterministic:fsa:bisimulation:language:equivalence},
    % and \item the fact that the reachability FSA and the specification FSA are deterministic and trimmed,
\end{iiilist}
it suffices to show that the two automata accept the same language to prove our claim. This, in turn, means that the {\Declare} specification returned by \cref{alg:imperative-to-declarative} accepts all and only the runs of the input safe and sound Workflow net. We now proceed to formally express our claim.

\begin{theorem}\label{thm:trace-equivalence}
    Given a safe and sound Workflow net {\Wn}, \cref{alg:imperative-to-declarative} returns a {\Declare} specification $\DecSpec$ such that: 
    \begin{iiilist}
	    \item any run of {\Wn} satisfies $\DecSpec$, and 
		\item any trace satisfying $\DecSpec$ is a run of {\Wn}.
    \end{iiilist}
\end{theorem}

\def\Instant {\ensuremath{\ell}}
\begin{proof}\vspace{-2ex}
    We prove that {\DecSpec} and {\Wn} satisfy the two conditions stated in the claim.

    \begin{asparaitem}
        \item[\textit{(i)}] Let $\PnFireSeq$
        %\in \PnSetFireSeq(\Wn)$ \todo{Undefined} 
        be a run of {\Wn}. We show that $\PnFireSeq \models \FormulaOf{\DecSpec}$. As {\Wn} is a Workflow net, it has a unique input place and a unique output place. Let $\Place_\letteri$ be $\InPl$ and $\Place_\lettero$ be $\OutPl$. % which are the input and output places of the net, i.e., such that the preset and the postset are empty, respectively. 
        In $\FormulaOf{\DecSpec}$, we have only one constraint for the templates $\End{\paramx}$ and $\AtMo{\paramx}$, namely $\End{\Pre{\Place_\lettero}}$ and $\AtMo{\Post{\Place_\letteri}}$.
        Let 
        $\{\Tx{\lettero_1}, \ldots \Tx{\lettero_n}\}$ be the preset of $\Place_\lettero$  (with $n \in \mathbb{N}$).
        Any run of {\Wn} must satisfy $\ltlalws\ltlevtly\left( \Tx{\lettero_1} \vee \ldots \vee \Tx{\lettero_n} \right)$, i.e., $\FormulaOf{\End{\Pre{\Place_\lettero}}}$, as one of the transitions in the preset of $\Place_\lettero$ must fire last.
        Let 
        $\{\Tx{\letteri_1}, \ldots \Tx{\letteri_k}\}$ be the postset of $\Place_\letteri$ (with $k \in \mathbb{N}$).
		No other place $\Place' \neq \Place_\letteri$ can be such that $\Tx{\letteri} \in \Post{\Place'}$ for any $\Tx{\letteri} \in \Post{\Place_\letteri}$, otherwise $\Tx{\letteri}$ would be a dead transition (thus contradicting soundness): The initial marking assigns no token to $\Place'$, and no marking except the initial one assigns a token to $\Place_\letteri$.
        As a consequence, any run of {\Wn} must satisfy 
        $\ltlalws \left( \left( \Tx{\letteri_1} \vee \ldots \vee \Tx{\letteri_k} \right) \to \neg \ltlnext\ltlevtly \left( \Tx{\letteri_1} \vee \ldots \vee \Tx{\letteri_k} \right) \right)$, i.e., $\FormulaOf{\AtMo{\Post{\Place_\letteri}}}$.
		\newline
        It remains to show
        that $\PnFireSeq \models \FormulaOf{\AltPrecShort{\Pre\Place}{\Post\Place}}$ for any arbitrary place $\Place \in \Places \setminus \{ \Place_\letteri, \Place_\lettero \}$.
        Assume by contradiction that $\PnFireSeq \not\models  \ltlalws \left( \Post\Place \limply \ltlyday( \lnot\Post\Place \ltlsince \Pre\Place ) \right) $. Then, there must be a timestep $\Instant < |\PnFireSeq|$ such that $\PnFireSeq, \Instant \models \Post\Place$, i.e., a transition in $\Post\Place$ was fired, but $\PnFireSeq, \Instant \not\models \ltlyday (\lnot \Post\Place \ltlsince \Pre\Place)$. 
        Notice that, as a transition in $\Post\Place$ was fired, it means that a transition in $\Pre\Place$ was fired at a timestep $\Instant' < \Instant$, as otherwise there would be no token assigned to $\Place$ at timestep $\Instant$. Then, for $\PnFireSeq, \Instant \not\models \ltlyday (\lnot \Post\Place \ltlsince \Pre\Place)$ to be true, it must be the case that a transition in $\Post\Place$ was fired at some timestep $\Instant''$ such that $\Instant' < \Instant'' < \Instant$, and no transition in $\Pre\Place$ has been fired between timesteps $\Instant''$ and $\Instant$. However, this, in conjunction with the fact the Workflow net is safe, implies that it would not have been possible to fire a transition in $\Post\Place$ at timestep $\Instant$: $\Place$ has no token assigned at timestep $\Instant$ as it was consumed to fire a transition in $\Post\Place$ at timestep $\Instant''$. Therefore, $\PnFireSeq \models \FormulaOf{\AltPrecShort{\Pre\Place}{\Post\Place}}$ for any arbitrary place $\Place \in \Places \setminus \{ \Place_\letteri, \Place_\lettero \}$, thus implying that $\PnFireSeq \models \FormulaOf{\DecSpec}$. %, as we wanted to prove.

        % \todo[inline]{{\normalsize Warning: I think we have a ``small'' problem here.
        
        % {
        %     \centering
        %     \includegraphics[width=1\linewidth]{notes/counterexample_algorithm.jpg}
        % }

        % If I'm not mistaken, the sequence $(a, b, d, c, b, b)$ should be a run of the Workflow net above, but it doesn't satisfy $\AltPrecShort{\bullet 2}{2 \bullet}$ because we can fire $b$ twice consecutively after firing $c$ and $d$. If this is a real counterexample, then I think a way we can fix this is by considering only nets in which any place can be marked by at most one token at any time. Maybe there's a way to fix this even if we assume that the number of maximum tokens (at any given time in a place) is bounded, but I'm not sure yet (I have doubts about this because neither {\LTL} nor {\LTLf} can ``count'').

        % A (possible) problem with this Workflow net: I don't think there's a way of having a run of this net that ends in a marking that assigns more than 0 tokens *only* to the output place. This means that this Workflow net doesn't have the ``proper completion'' property, right? Maybe we can prove the claim by assuming that the Workflow net is sound.}
        % CDC: Great catch, Giovanni! But indeed, sound WF nets are a subclass of sound Petri nets, and sound Petri nets are 1-bounded (a.k.a.\ ``safe'', which means, no place can be marked with more than 1 in any reachable state). Your observation about a missed proper completion is fully correct anyway!}
        
        \item[\textit{(ii)}] We now show that if a trace $\PnFireSeq$ is such that $\PnFireSeq \models \FormulaOf{\DecSpec}$ then $\PnFireSeq$ is a run of $\Wn$. Let $\Place_\letteri$ be $\InPl$ and $\Place_\lettero$ be $\OutPl$ again. Since $\PnFireSeq \models \FormulaOf{\DecSpec}$, we have that the trace correctly ends with a transition in the preset of $\Place_\lettero$, because $\PnFireSeq \models \FormulaOf{\End{\Pre{\Place_\lettero}}}$. Also, in a run of \Wn, the transitions in $\Post{\Place_\letteri}$ can only be fired once, otherwise $\Pre{\Place_\letteri}$ would be non-empty against the definition of $\InPl$. This holds true in $\PnFireSeq$, as $\PnFireSeq \models \FormulaOf{\AtMo{\Post{\Place_\letteri}}}$. 
        Notice that, unlike all other transitions in \Wn, only those in $\Post{\Place_\letteri}$ do \emph{not} map to $x$ for $\AltPrecShort{y}{x}$ in \DecSpec by design of \cref{alg:imperative-to-declarative}. Therefore, every trace will begin with the occurrence of one of the transitions in $\Post{\Place_\letteri}$ as it happens with the runs of \Wn. % Furthermore, as the initial marking $\Marking_0$ assigns only one token to $\InPl$, it must be the case that the first transition fired is in $\Post\InPl$.\todo{This seems unnecessary to the overall discourse}
        
        It remains to show that every transition in the trace $\PnFireSeq$ was fired in $\Wn$ following the preceding sequence of transitions in the trace. % Note that this also implies that the last marking reached in the trace is the final marking, as $\PnFireSeq \models \End{\Pre\OutPl}$ implying that the output place is marked with a token following the last timestep of $\PnFireSeq$. \todo{This seems unnecessary to the overall discourse}
        Suppose by contradiction that this is not the case, i.e., that there is some transition $\Transition$ fired at a timestep $\Instant$ which could not have been fired in {\Wn} given the prefix of $\PnFireSeq$ from $1$ to $\Instant -1$. Then, this implies that at least one of the places $\Place$ such that $\Transition \in \Post\Place$ does not have a token at timestep $\Instant -1$, i.e., $\Marking_{\Instant-1}(\Place)=0$. Two conditions can entail this situation: Either no transition in $\Pre\Place$ was fired before, or a transition in $\Post\Place$ was fired since the last timestep in which a transition in $\Pre\Place$ was fired, consuming the only token assigned to $\Place$. Both cases contradict the fact that $\PnFireSeq \models \FormulaOf{\AltPrecShort{\Pre\Place}{\Post\Place}}$, thus proving that $\PnFireSeq$ is a valid sequence of transitions with respect to $\Wn$. Thus, $\PnFireSeq$ is a run of $\Wn$.
    \end{asparaitem}
    
\end{proof}
Given \cref{obs:deterministic:fsa:bisimulation:language:equivalence}, we immediately obtain the following corollary.
\begin{corollary}
    The {\Declare} specification $\DecSpec$ given in output by \cref{alg:imperative-to-declarative} is bisimilar to the input safe and sound Workflow net {\Wn}.
\end{corollary}
%
% Showing indeed that \cref{alg:imperative-to-declarative} produces a {\Declare} specification that behaves in an indistinguishable manner from the input sound Workflow net {\Wn}.
%
This result has a profound implication that transcends {\Declare} and {\LTLf} but pertains to the languages recognized by safe and sound Workflow nets.
\begin{theorem}
	Languages of safe and sound Workflow nets are star-free regular expressions.
\end{theorem}
\begin{proof}\vspace{-1ex}
	The claim follows from \cref{thm:trace-equivalence}, recalling that {\Declare} patterns are expressed in {\LTLf}, which is expressively equivalent to star-free regular expressions~\cite{DeGiacomo.Vardi/IJCAI2013:LDLf}.
\end{proof}%
%\todo{@Marco: We wrote ``Notice that this result is established when focusing on the alphabet of transitions in the Workflow net (as done here), but does not hold if one considers the alphabet of labels in a labeled Workflow net wherein multiple transitions can map to the same one.'' But I removed it, because it would backfire in light of what we stated after \cref{def:workflownet} and in \cref{sec:algo}}

% \textbf{Observation}: the strong bisimulation property we obtain is between the reachability graphs (of the markings) of the input Workflow net and the reachability graph of the Petri net associated to the \Declare{} specification (this is the one from Claudio's paper \Declare{} to Imperative) given in output by the algorithm. \textbf{Sanity check}: is the reachability graph of a Petri net a deterministic LTS (states: markings; actions: transitions of the Petri net; transition relation: accessibility relation of the reachability graph; initial state: initial marking of the Petri net)? Obviously yes: each transition modifies a marking in a deterministic way. This means we have to rewrite a bit this part: what we are showing is the run equivalence between the reachability graphs of the input Workflow net and the Petri net obtained from the \Declare specification.

\noindent\textbf{Space and time complexity.} \label{sec:algo:subsec:spacetime} \Cref{alg:imperative-to-declarative} outputs a {\Declare} specification $\DecSpec = \left( \DecRepertoire, \AuAlph, \CnsSet \right)$ which contains, for each place in the input Workflow net $\Wn = \left( \Places, \Transitions, \FlowRel \right)$, a constraint with the pre- and post-sets as its actual parameters. Each transition that is in relation with a  place $\Place$ in the flow relation $F$ appears exactly once in the constraint stemming from $\Place$; therefore, the space complexity class of \cref{alg:imperative-to-declarative} is $\mathcal{O}\left( |\FlowRel| \right)$. %, $\Theta({F})$. In the worst case, % every such constraint is exerted on all transitions of the Workflow net. Therefore, 
%the space complexity is thus $\mathcal{O}\left( |\Places| \times |\Transitions|\right)$.\todo{[alternative complexity argument] Each transition that is in relation with the place in the flow relation appears exactly once in the constraint of the place; therefore, the space complexity is $\mathcal{O}\left( |\FlowRel| \right)$. ACCEPTED!}
As for the time complexity, we can assume that a pre-processing step is conducted to represent $F$ in the form of a sequence of pairs, associating every place to its pre-set and post-set. The cost of this operation is $\Theta(F)$ and $\mathcal{O}\left( |\Places| \times |\Transitions|\right)$. %it is possible to access the preset and postset of any place in $\mathcal{O}(1)$; note how this is possible by pre-processing the {\Wn} into a data structure that maps each place to a pair of sets of transitions, one containing the transition in the place's preset and the other the transition in the place's postset. 
For each place, the algorithm performs up to three if-checks %, as it only needs to access the place's preset and postset; after the if-check, the algorithm creates a 
and then a new constraint is created in constant time, hence $\mathcal{O}(|P|)$. % which is also done in $\mathcal{O}\left(1\right)$ time, since the constraint is defined with respect to the transitions in the preset and postset of the place. 
Therefore, the time complexity of the algorithm is bounded by the update of $\CnsSet$, necessitating up to $\mathcal{O}(|P|\times|T|)$ time.

%\todo{CDC@Marco: underline that not only are we polynomial in the size of the net, but also that a brute-force approach would not guarantee soundness. Once a pair of parallel transitions are turned into interleaving sequences of transitions in the FSA, they become indistinguishable to {\Declare}. See $\Tx\taske$ and $\Tx\taskf$ in \cref{fig:wnexample,fig:reachabilitygraph,fig:automaton}.}

Next, we experimentally validate and put the above theoretical results to the test.

\section{Implementation and evaluation}
\label{sec:evaluation}
We implemented \cref{alg:imperative-to-declarative} in the form of a proof-of-concept prototype encoded in Python. %, which supports the transformation of a safe and sound Workflow net 
%represented in a pnml file, 
% into a corresponding declarative specification. %, fully compatible with the target-branched version of MINERful proposed in \cite{DBLP:conf/bpm/CiccioMM14}. 
The tool, testbeds, and experimental results are available for public access.\footnote{\url{https://github.com/l2brb/Sp3llsWizard}\label{code:repo}} 
%In the following, 
In the following, we report on tests conducted with our algorithm's implementation to empirically confirm its soundness, assess memory efficiency, and gauge runtime performance. % First, we perform a bisimulation test between the automata derived from the declarative specification and the one obtained from the reachability graph of the workflow net. Next, we evaluate memory usage and time consumption to provide insights into the efficiency and scalability of the proposed algorithm under changing input workflow net configurations. 
Finally, we showcase a process diagnostic application as a downstream task for our approach. % declarative specifications returned by the algorithm on a real-world event log. 

\subsection{Automata bisimulation}
\label{sec:automatabisim}
\begin{figure}[tb]
    \centering
    \resizebox{\textwidth}{!}{
        \begin{tikzpicture}[->,>=stealth',shorten >=1pt,node distance=2cm,
                    thick,
                    initial text = {},
                    state/.style = {
                        circle,
                        draw,
                        inner sep=1mm,
                    },]
                            
      \node[initial, state] (q0) {$\AuSt_0$};
      \node[state] (q1) [right of=q0] {$\AuSt_1$};
      \node[state] (q2) [right of=q1] {$\AuSt_2$};
      \node[state] (q3) [right of=q2] {$\AuSt_3$};
      \node[state] (q4) [right of=q3] {$\AuSt_4$};
      \node[state] (q5) [right of=q4] {$\AuSt_5$};
      \node[state] (q6) [above=0.25cm of q5] {$\AuSt_6$};
      \node[state] (q7) [right of=q5] {$\AuSt_7$};
      \node[state] (q8) [right of=q7] {$\AuSt_8$};
      \node[accepting, state] (q9) [right of=q8] {$\AuSt_9$};
      
      \path
        (q0) edge node [midway, fill=white] {$t_{\taska}$} (q1)
        
        (q1) edge node [midway, fill=white] {$t_{\taskb}$} (q2)
        
        (q2) edge [bend left=30] node [midway, fill=white] {$t_{\taskc}$} (q3)
        
        (q2) edge [bend right=30] node [midway, fill=white] {$t_{\taskd}$} (q3)
        
        (q3) edge node [midway, fill=white] {$t_{\taske}$} (q4)
        
        (q4) edge node [midway, fill=white] {$t_{\taskg}$} (q5)
        
        (q4) edge [bend left=15] node [midway, fill=white] {$t_{\taskf}$} (q6)
        
        (q6) edge [bend left=15] node [midway, fill=white] {$t_{\taskg}$} (q7)
        
        (q5) edge node [midway, fill=white] {$t_{\taskf}$} (q7)
        
        (q7) edge node [midway, fill=white] {$t_{\tasku}$} (q8)
        
        (q8) edge node [midway, fill=white] {$t_{\taskv}$} (q9)
        
        (q8) edge [bend left=15] node [midway, fill=white] {$t_{\taskw}$} (q1);
     
\end{tikzpicture}%
    }\vspace{-2ex}
    \caption[FSA equivalent to the specification]{FSA of the specification in \cref{tab:constraints}}
	\label{fig:automaton}
\end{figure}
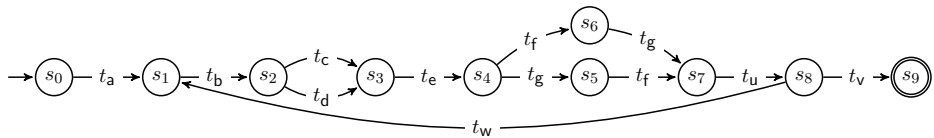
To experimentally validate the correctness of the implementation of \cref{alg:imperative-to-declarative}, we performed a preliminary comparison of the reachability FSA (\cref{def:reachability-graph}) of known Workflow nets and the specification FSA (\cref{def:specification-fsa}) consisting of the {\Declare} constraints returned by our tool.
\Cref{fig:automaton} illustrates the FSA of the specification derived from the Workflow net in \cref{fig:wnexample}, computed with a dedicated module presented in~\cite{DiCiccio.etal/IS2017:ResolvingInconsistenciesRedundanciesDeclare}. Also by visual inspection, we can conclude that the two FSAs are bisimilar, as expected. Owing to space constraints, we cannot portray the entire range of automata derived from the Workflow nets in our experiments. The interested reader can find the full collection collection (including non-free choice nets such as that of~\cite[Fig.~24]{DBLP:journals/jcsc/Aalst98}) in our public codebase.\textsuperscript{\ref{code:repo}}

\subsection{Performance analysis}
\label{sec:performance}
Here, we report on the quantitative assessment of our solution in terms of scalability given an increasing workload, and against real-world testbeds. For the former, we observe the time and space performance of our implemented prototype fed in input with Workflow nets of increasing size. We control the expansion process in two directions, so as to obtain the following separate effects:
\begin{iiilist}
	\item more constraints are generated, while each is exerted on up to three literals;
	\item the amount of generated constraints remains fixed, while the literals mapped to their parameters increase.
\end{iiilist}
For the real-world testbed, we take as input processes discovered by a well-known imperative process mining algorithm from a collection of openly available event logs.
%\todo{CDC: Added passage here. Reason: Do not just say what we did. MOTIVATE first, then explain.}
% through the construction of synthetic Workflow nets, resulting in a progressively more complex output configuration. We validated the execution and extend testing on process models derived from a collection of real-world event logs. 
We conducted the performance tests on an AMD Ryzen 9 8945HS CPU at 4.00 GHz with 32 GB RAM running Ubuntu 24.04.1. %To ensure results reliability, we 
For the sake of reliability, we ran three iterations for every test configuration and averaged the outputs to derive the final result.
%
\begin{comment}
\begin{wrapfigure}[4]{r}{0.3\textwidth}
	% [9] = numero di linee di testo da far scorrere accanto
	%  r   = posizionamento a destra (l per sinistra)
	%  0.4\textwidth = larghezza del riquadro
	\vspace{-0.7em} % eventualmente regola lo spazio verticale
	\centering
	\resizebox{0.7\textwidth}{!}{%
		\input{figures/workflow-nets/rsample} % Il tuo file .tex con TikZ
		\caption{Base net used to \\apply the expansion rules \label{fig:rsample}}
	}
	\vspace{-0.72em}
\end{wrapfigure}
\end{comment}
%
% BASE NETS
%
\begin{figure}[tb]
    \centering
    \begin{floatrow}
        \ffigbox{
            \centering
            \resizebox{0.2\textwidth}{!}{
                \begin{tikzpicture}[>=stealth',x=3em,y=2.5em]
\begin{scriptsize}
\tikzstyle{transition}=[rectangle,thick,draw=black!75,minimum height=2em,minimum width=2em,text width=1em,align=center,fill=LimeGreen,fill opacity=0.5]
\tikzstyle{newtransition}=[rectangle,thick,draw=black!75,minimum height=2em,minimum width=2em,text width=1em,align=center]
\tikzstyle{silenttransition}=[rectangle,thick,fill=black,minimum height=2em]
\tikzstyle{outputplace}=[double,draw,circle,minimum height=2em,fill=Salmon,fill opacity=0.7]

\node [place,fill=Salmon,fill opacity=0.7] (0) at (0, 0) [label=below:{$p_1$},tokens=1] {};

\node [outputplace] (1) at ( +2, 0)  [label=below:{$p_2$}]          {};

\node [transition,label=below:{$t_1$}] (a) at ( +1,  0) {}
edge [pre]                 (0)
edge [post]                (1);
\end{scriptsize}
\end{tikzpicture}
            }
        }{
            \caption[Base net used to apply the expansion rules]{Base net used to initiate the expansion mechanism in \cref{fig:Trules}}
            \label{fig:rsample}
        }
        \ffigbox{
            \centering
            \raisebox{-2.5em}{ % <-- modifica la distanza in base al risultato desiderato
                \resizebox{0.48\textwidth}{!}{
                    \begin{tikzpicture}[>=stealth',x=1.125cm,y=0.75cm]%,bend angle=45,auto,bend angle=90]
\begin{scriptsize}
%  \tikzstyle{place}=[circle,thick,minimum size=6mm]
\tikzstyle{transition}=[rectangle,thick,draw=black!75,minimum height=2em,minimum width=2em,text width=1em,align=center,fill=LimeGreen,fill opacity=0.5]
\tikzstyle{newtransition}=[rectangle,thick,draw=black!75,minimum height=2em,minimum width=2em,text width=1em,align=center]
\tikzstyle{silenttransition}=[rectangle,thick,fill=black,minimum height=2em]
\tikzstyle{outputplace}=[double,draw,circle,minimum height=2em,fill=Salmon,fill opacity=0.7]

\node [place,fill=Salmon,fill opacity=0.7] (0) at (0, 0) [label=below:{$p_1$},tokens=1] {};

%\node [place,fill=Salmon,fill opacity=0.7] (0) at (0, 0) [label=above:{$\InPl$},tokens=1] {};

\node [place] (2) at (+2, 0) [] {};
\node [place] (3) at (+4, 0) [] {};

%\node [outputplace] (1) at ( +6, 0)  [label=above:{$\OutPl$}]          {};
\node [outputplace] (1) at ( +6, 0)  [label=below:{$p_2$}]          {};

\node [transition,label=below:{$t_1$}] (a) at ( +3,  -0.5) {}
edge [pre]                 (2)
edge [post]                (3);

\node [newtransition] (b) at ( +1,  0) {}
edge [pre]                 (0)
edge [post]                (2);

\node [newtransition] (b) at ( +5,  0) {}
edge [pre]                 (3)
edge [post]                (1);

\node [newtransition] (d) at ( +3,  0.5) {}
edge [pre]                 (2)
edge [post]                (3);
\end{scriptsize}
\end{tikzpicture}
                }
            }
        }{
            \caption[Conditional expansion rule applied to increase constraints dimensionality]{Conditional expansion of the net in \cref{fig:T1a}}
            \label{fig:rconditional}
        }
    \end{floatrow}
\end{figure}%
% EXPANSION RULES
%
\begin{figure}[tb]
	\centering
	\begin{subfigure}[t]{0.45\linewidth}
		\centering
		\resizebox{\linewidth}{!}{
			\begin{tikzpicture}[>=stealth',x=3em,y=2.5em]%,bend angle=45,auto,x=1.125cm,y=0.75cm,bend angle=90]
\begin{scriptsize}
%  \tikzstyle{place}=[circle,thick,minimum size=6mm]
\tikzstyle{transition}=[rectangle,thick,draw=black!75,minimum height=2em,minimum width=2em,text width=1em,align=center,fill=LimeGreen,fill opacity=0.5]
\tikzstyle{newtransition}=[rectangle,thick,draw=black!75,minimum height=2em,minimum width=2em,text width=1em,align=center]
\tikzstyle{silenttransition}=[rectangle,thick,fill=black,minimum height=2em]
\tikzstyle{outputplace}=[double,draw,circle,minimum height=2em,fill=Salmon,fill opacity=0.7]

\node [place,fill=Salmon,fill opacity=0.7] (0) at (0, 0) [label=below:{$p_1$},tokens=1] {};
\node [outputplace] (1) at ( +6, 0)  [label=below:{$p_2$}]          {};

\node [place] (2) at (+2, 0) [] {};
\node [place] (3) at (+4, 0) [] {};

\node [transition,label=below:{$t_1$}] (a) at ( +3,  0) {}
edge [pre]                 (2)
edge [post]                (3);

\node [newtransition] (b) at ( +1,  0) {}
edge [pre]                 (0)
edge [post]                (2);

\node [newtransition] (c) at ( +5,  0) {}
edge [pre]                 (3)
edge [post]                (1);

% \node [newtransition] (d) at ( +3,  1) {}
% edge [pre]                 (2)
% edge [post]                (3);
\end{scriptsize}
\end{tikzpicture}
		}
		\caption[Sequential expansion]{Sequential expansion of the net in \cref{fig:rsample}}
		\label{fig:T1a}
	\end{subfigure}%
	\hfill
	\begin{subfigure}[t]{0.45\linewidth}
		\centering
		\resizebox{\linewidth}{!}{
			\tikzset{
    -|/.style={to path={-| (\tikztotarget)}},
    |-/.style={to path={|- (\tikztotarget)}},
}

\begin{tikzpicture}[>=stealth',x=3em,y=2.5em]%,bend angle=45,auto,x=1.125cm,y=0.75cm,bend angle=90]
\begin{scriptsize}
%  \tikzstyle{place}=[circle,thick,minimum size=6mm]
\tikzstyle{transition}=[rectangle,thick,draw=black!75,minimum height=2em,minimum width=2em,text width=1em,align=center,fill=LimeGreen,fill opacity=0.5]
\tikzstyle{newtransition}=[rectangle,thick,draw=black!75,minimum height=2em,minimum width=2em,text width=1em,align=center]
\tikzstyle{silenttransition}=[rectangle,thick,fill=black,minimum height=2em]
\tikzstyle{outputplace}=[double,draw,circle,minimum height=2em,fill=Salmon,fill opacity=0.7]

\node [place,fill=Salmon,fill opacity=0.7] (0) at (0, 0) [label=below:{$p_1$},tokens=1] {};
\node [outputplace] (1) at ( +6, 0)  [label=below:{$p_2$}]          {};

\node [place] (2) at (+2, 0) [] {};
\node [place] (3) at (+4, 0) [] {};

\node [place] (4) at (+2, -1.2) [] {};
\node [place] (5) at (+4, -1.2) [] {};

\node [transition,label=below:{$t_1$}] (a) at ( +3,  0) {}
edge [pre]                 (2)
edge [post]                (3);

\node [newtransition] (b) at ( +1,  0) {}
edge [pre]                 (0)
edge [post]                (2) 
edge [post,|-]             (4);

\node [newtransition] (c) at ( +5,  0) {}
edge [pre]                 (3)
edge [post]                (1)
edge [pre,|-]              (5);

\node [newtransition] (d) at ( +3,  -1.2) {}
edge [pre]                 (4)
edge [post]                (5);
\end{scriptsize}
\end{tikzpicture}
		}
		\caption[Parallel expansion]{Parallel expansion of the net in \cref{fig:T1a}}
		\label{fig:T3a}
	\end{subfigure}
	\begin{subfigure}[t]{0.5\linewidth}
		%\centering
		\resizebox{0.9\linewidth}{!}{
			\tikzset{
    -|/.style={to path={-| (\tikztotarget)}},
    |-/.style={to path={|- (\tikztotarget)}},
}

\begin{tikzpicture}[>=stealth',x=3em,y=2.5em]%,bend angle=45,auto,x=1.125cm,y=0.75cm,bend angle=90]
\begin{scriptsize}
%  \tikzstyle{place}=[circle,thick,minimum size=6mm]
\tikzstyle{transition}=[rectangle,thick,draw=black!75,minimum height=2em,minimum width=2em,text width=1em,align=center,fill=LimeGreen,fill opacity=0.5]
\tikzstyle{newtransition}=[rectangle,thick,draw=black!75,minimum height=2em,minimum width=2em,text width=1em,align=center]
\tikzstyle{silenttransition}=[rectangle,thick,fill=black,minimum height=2em]
\tikzstyle{outputplace}=[double,draw,circle,minimum height=2em,fill=Salmon,fill opacity=0.7]

\node [place,fill=Salmon,fill opacity=0.7] (0) at (0, 0) [label=below:{$p_1$},tokens=1] {};
\node [outputplace] (1) at ( +6, 0)  [label=below:{$p_2$}]          {};

\node [place] (2) at (+2, 0) [] {};
\node [place] (3) at (+4, 0) [] {};

\node [place] (4) at (+2, -1.2) [] {};
\node [place] (5) at (+4, -1.2) [] {};

\node [transition,label=below:{$t_1$}] (a) at ( +3,  0) {}
edge [pre]                 (2)
edge [post]                (3);

\node [newtransition] (b) at ( +1,  0) {}
edge [pre]                 (0)
edge [post]                (2) 
edge [post,|-]             (4);

\node [newtransition] (c) at ( +5,  0) {}
edge [pre]                 (3)
edge [post]                (1)
edge [pre,|-]              (5);

\node [newtransition] (d) at ( +3,  -1.2) {}
edge [pre]                 (4)
edge [post]                (5);

\node [newtransition] (e) at ( +3,  +1.2) {}
edge [pre,-|]                 (0)
edge [post,-|]                (1);
\end{scriptsize}
\end{tikzpicture}
		}
		\caption[Conditional expansion]{Conditional expansion of the net in \cref{fig:T3a}}
		\label{fig:T2a}
	\end{subfigure}%
	\hfill%
	\begin{subfigure}[t]{0.49\linewidth}
		\centering
		\resizebox{\linewidth}{!}{
			\tikzset{
    -|/.style={to path={-| (\tikztotarget)}},
    |-/.style={to path={|- (\tikztotarget)}},
}

\tikzset{
    -|/.style={to path={-| (\tikztotarget)}},
    |-/.style={to path={|- (\tikztotarget)}},
}

\begin{tikzpicture}[>=stealth',x=3em,y=2.5em]%,bend angle=45,auto,x=1.125cm,y=0.75cm,bend angle=90]
\begin{scriptsize}
%  \tikzstyle{place}=[circle,thick,minimum size=6mm]
\tikzstyle{transition}=[rectangle,thick,draw=black!75,minimum height=2em,minimum width=2em,text width=1em,align=center,fill=LimeGreen,fill opacity=0.5]
\tikzstyle{newtransition}=[rectangle,thick,draw=black!75,minimum height=2em,minimum width=2em,text width=1em,align=center]
\tikzstyle{silenttransition}=[rectangle,thick,fill=black,minimum height=2em]
\tikzstyle{outputplace}=[double,draw,circle,minimum height=2em,fill=Salmon,fill opacity=0.7]

\node [place,fill=Salmon,fill opacity=0.7] (0) at (-2, 0) [label=below:{$\Place_1$},tokens=1] {};
\node [outputplace] (1) at ( +8, 0)  [label=below:{$\Place_2$}]          {};

\node [place,fill=SkyBlue,fill opacity=0.4] (2) at (+2, 0) [label=above:{$\Place_1'$}] {};
\node [place,fill=SkyBlue,fill opacity=0.4] (3) at (+4, 0) [label=above:{$\Place_2'$}] {};

\node [place] (4) at (+2, -1.2) [] {};
\node [place] (5) at (+4, -1.2) [] {};

\node [place] (6) at (0, 0)     [] {};
\node [place] (7) at (+6, 0)     [] {};

\node [transition,label=below:{$t_1$}] (a) at ( +3,  0) {}
edge [pre]                 (2)
edge [post]                (3);

\node [newtransition] (b) at ( +1,  0) {}
edge [pre]                 (6)
edge [post]                (2) 
edge [post,|-]             (4);

\node [newtransition] (c) at ( +5,  0) {}
edge [pre]                 (3)
edge [post]                (7)
edge [pre,|-]              (5);

\node [newtransition] (d) at ( +3,  -1.2) {}
edge [pre]                 (4)
edge [post]                (5);

\node [newtransition] (e) at ( +3,  +1.2) {}
edge [pre,-|]                 (0)
edge [post,-|]                (1);

\node [newtransition] (e) at ( -1,  0) {}
edge [pre]                 (0)
edge [post]                (6);

\node [newtransition] (e) at ( +7,  0) {}
edge [pre]                 (7)
edge [post]                (1);

\node [newtransition] (e) at ( +3,  -2.4) {}
edge [pre,-|]                 (7)
edge [post,-|]                (6);
\end{scriptsize}
\end{tikzpicture}
		}
		\caption[Loop expansion]{Loop expansion of the net in \cref{fig:T2a}}
		\label{fig:T4a}
	\end{subfigure}
	\caption[Transformation rules]{Transformation rules used to iteratively expand a safe and sound Workflow net.}
	\label{fig:Trules}
\end{figure}
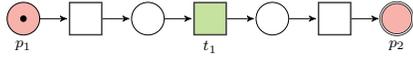
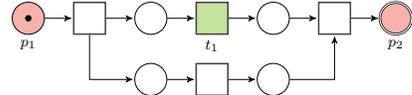
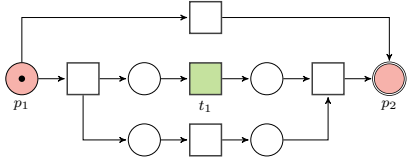
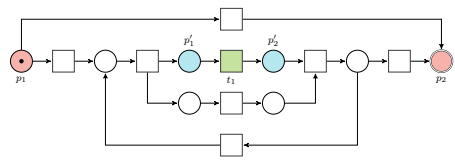

\begin{figure}[tb]
	\centering
	\begin{subfigure}[t]{0.48\linewidth}
		\centering        \includegraphics[width=\linewidth]{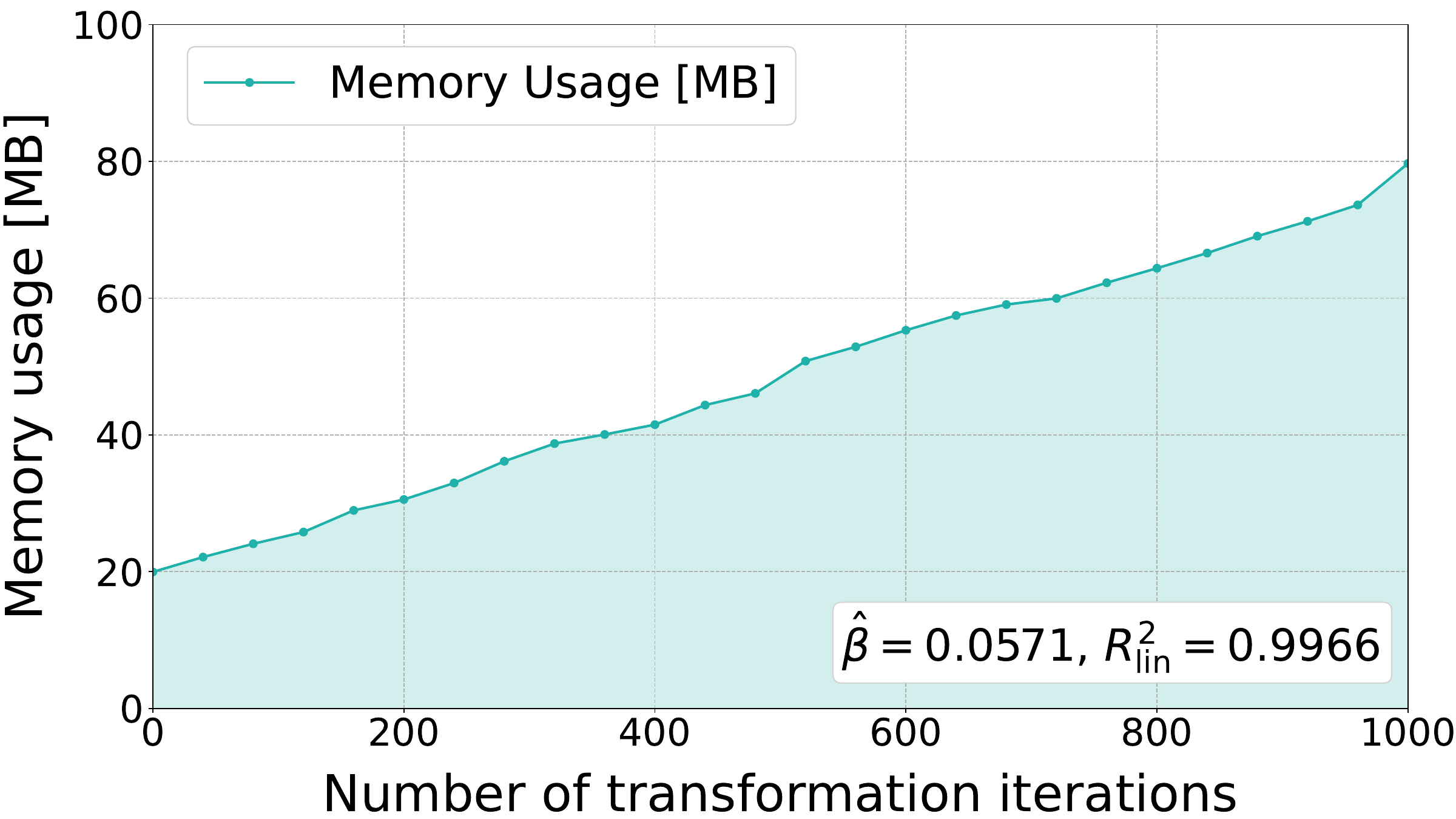}
		\caption[Memory usage plot for an incremental number of constraints]{Memory usage}
		\label{fig:memoryusage-n}
	\end{subfigure}%
	\begin{subfigure}[t]{0.48\linewidth}
		\centering
		\includegraphics[width=\linewidth]{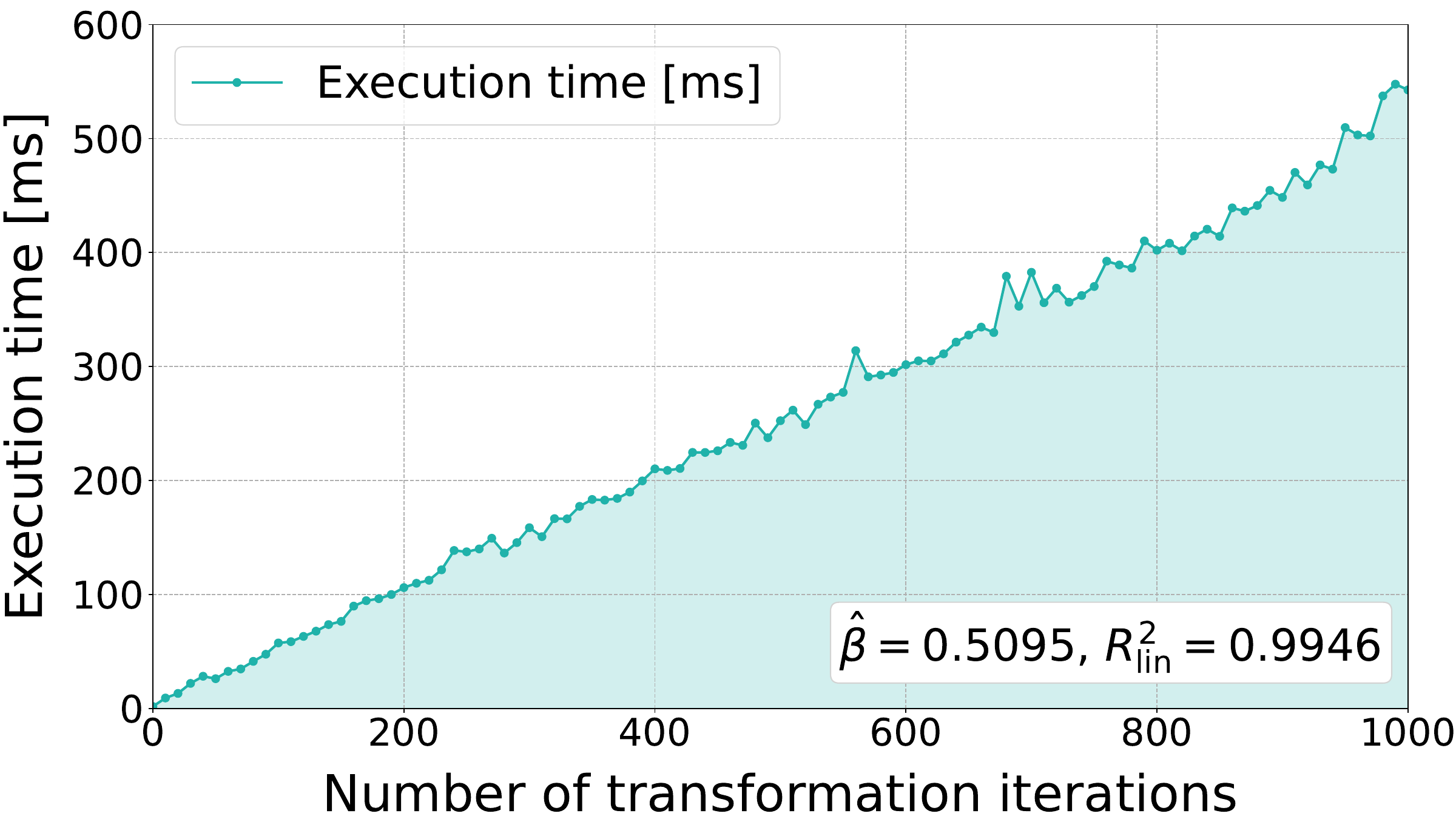}
		\caption[Execution time plot for an incremental number of constraints]{Execution time}
		\label{fig:executiontime-n}
	\end{subfigure}%
	\caption[Test results for incremental number of constraints]{Test results for the incremental number of constraints setup}
	\label{fig:nconstraints}
\end{figure}
\noindent\textbf{Increasing constraint-set cardinality.} \label{sec:evaluation:subsec:nconsraints}To examine the effectiveness of the \cref{alg:imperative-to-declarative} in handling an incremental number of constraints, we examine memory utilization and execution time through the progressive rise in the complexity of the input Workflow net. Our evaluation method relies on an expansion mechanism that iteratively applies a structured pattern of four soundness-preserving transformation rules from~\cite{Aalst/ICATPN97:VerificationofWfNs} to progressively increase the number of nodes and their configuration. This leads to a gradual increase in the number of constraints our algorithm needs to initiate. 
Starting from the Workflow net in \cref{fig:rsample}, we designate transition $t_1$ as a fixed ‘pivot’, retaining the initial and final places ($\Place_1$ and $\Place_2$), and iteratively apply the expansion mechanism illustrated in \cref{fig:Trules}.
We apply known workflow patterns
%\todo{CDC@Luca: I am trying to fix all the occurrences where we use plurals as adjectives (like ``rules pattern'', which is incorrect: either ``rule pattern'' or ``rules' pattern''). Please make sure to get rid of these things should they recur somewhere.}
in the following order: (\cref{fig:T1a}) We add a transition before and after $\Transition_1$; (\cref{fig:T3a}) We introduce a parallel execution path; (\cref{fig:T2a}) We insert an exclusive branch; (\cref{fig:T4a}) Finally, we incorporate a loop structure. Upon completion of the expansion process, we execute the algorithm, record the results, and initiate a new iteration, maintaining $\Transition_1$ unchanged while reassigning $\Place_1$ and $\Place_2$ with the places that have $\Transition_1$ in the preset and postset (see the places colored in blue and labeled with $\Place_1'$ and $\Place_2'$ in \cref{fig:T4a}). We reiterated the procedure \num{1000} times.
\\
\indent \Cref{fig:nconstraints}
%\todo{L@himself: bigger labels}
displays the registered memory usage and execution time.
% of the output trends across different net configurations, 
To interpret the performance trends, we employ two well-established measures%~\cite{altman2015simplelinearregression}
: the coefficient of determination $R^2_\textrm{lin}$, which assesses the goodness-of-fit of the data to a linear trend,
and the $\hat{\beta}$ rate, which serves as a meter for the line's slope.
%
%\todo{CDC4Luca: Rephrased in a tentatively less generic form. Please check if these statements actually make sense!, cambio qui perchè non è formalmente corretto parlare di R square e correlazione, ma ho capito il senso :) }
%
As depicted in \cref{fig:memoryusage-n}, memory consumption increases linearly with the number of iterations of the expansion mechanism confirmed by $R^2_\textrm{lin} = 0.9966$ and a low slope increase ($\hat{\beta} =0.0571$). \Cref{fig:executiontime-n} displays the execution time plot, with $R^2_\textrm{lin} = 0.9946$ and $\hat{\beta} = 0.5095$, thus indicating a linear trend with a slope exhibiting a moderate incline. %less than  increasing workload. 
We remark that these results are in line with the theoretical analysis of the space and time complexity in \cref{sec:algo:subsec:spacetime}. 
\begin{figure}[tb]
	\centering
	\begin{subfigure}[t]{0.48\linewidth}
		\centering
		\includegraphics[width=\linewidth]{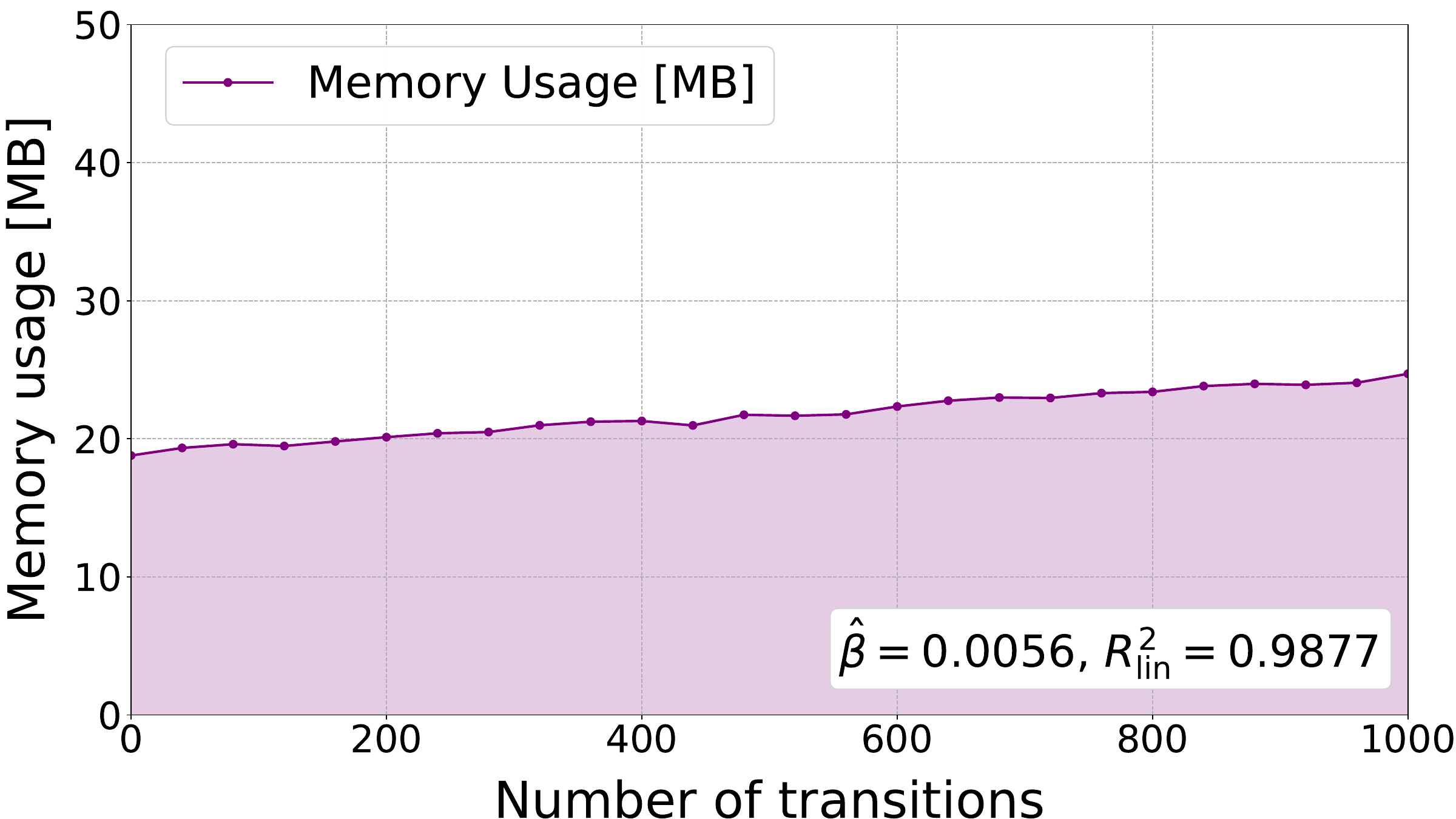}
		\caption[Memory usage for incremental constraints dimension]{Memory usage}
		\label{fig:memoryusage-d}
	\end{subfigure}%
	\begin{subfigure}[t]{0.48\linewidth}
		\centering
		\includegraphics[width=\linewidth]{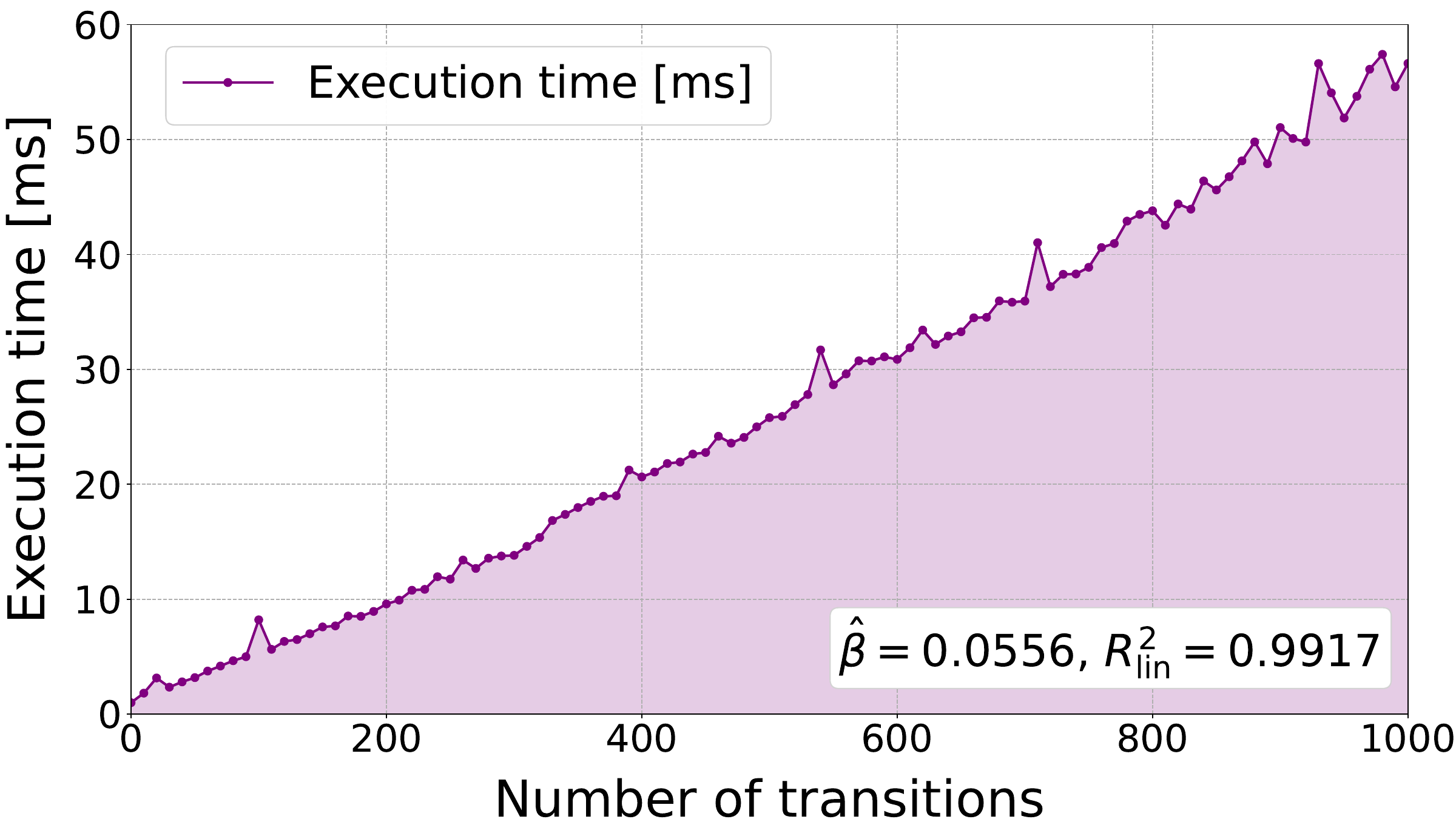}
		\caption[Execution time for incremental constraints dimension]{Execution time}
		\label{fig:executiontime-d}
	\end{subfigure}%
	\caption[Incremental constraints dimension tests]{Test results for the incremental constraints dimension setup}
	\label{fig:dconstraints}
\end{figure}
\\
\noindent\textbf{Increasing constraint formula size.} 
\label{sec:evaluation:subsec:dconsraints}
%\todo{Dimentionality does not seem the right word here: \url{https://dictionary.cambridge.org/dictionary/english/dimensionality} ----> Size? Changed.}
Here, we configure the test on memory usage and execution time to investigate the algorithm's performance while handling an expanding constraints' formula size (i.e., with an increasing number of disjuncts). To this end, we progressively broaden the Workflow net by applying the soundness-preserving conditional expansion rule from~\cite{Aalst/ICATPN97:VerificationofWfNs} depicted in \cref{fig:rconditional} to transition $\Transition_1$ in the net of \cref{fig:T1a}. We reiterate the process \num{1000} times. \Cref{fig:dconstraints} displays the results we registered.
Observing \cref{fig:memoryusage-d},
%\todo{L@himself: bigger labels}
we can assert that the memory utilization increases linearly % of the algorithm remains linear even by increasing the constraint dimensionality. In this context, the 
($R^2_\textrm{lin}=\num{0.9877}$) with a minimal rate ($\hat{\beta}=\num{0.0056}$). The execution time plotted in \cref{fig:executiontime-d} also exhibits a linear increase ($R^2_\textrm{lin} = 0.9917$), with a moderate slope inclination ($\hat{\beta} = 0.0556$). Once more, the results align with the theoretical complexity analysis in~\cref{sec:algo:subsec:spacetime}.
%\todo{Claudio's pass: till here.}
%
\begin{comment}
	\begin{wrapfigure}[6]{r}{0.4\textwidth}
		% [9] = numero di linee di testo da far scorrere accanto
		%  r   = posizionamento a destra (l per sinistra)
		%  0.4\textwidth = larghezza del riquadro
		%\vspace{0pt} % eventualmente regola lo spazio verticale
		\centering
		\resizebox{0.9\textwidth}{!}{%
			\input{figures/workflow-nets/rconditional} % Il tuo file .tex con TikZ
			\caption{Conditional expansion rule applied to increase constraint \\ dimensionality.}
		}
		\label{fig:rconditional}
	\end{wrapfigure}
	%
\begin{figure}[tb]
	\centering
	\resizebox{.5\textwidth}{!}{
		\input{figures/workflow-nets/rconditional}
	}
	\caption[Conditional expansion rule]{Conditional expansion rule applied to increase constraint dimensionality.}
	\label{fig:rconditional}
\end{figure}
\end{comment}

%REAL-WORLD TESTING
\noindent\textbf{Real-world process model testing.} \label{sec:evaluation:subsec:realeorld}
To evaluate the performance
%\todo{CDC: Watch out, we do not verify correctness here.}
of our algorithm in application on real process models, we conduct the same memory usage and execution time tests employing Workflow nets directly derived from a collection of real-life event logs available at \textit{4TU.ResearchData}.\footnote{The event logs used in our experiments are publicly available at \url{https://data.4tu.nl/}}
\begin{table}[tb]
    \centering
    \floatsetup[table]{capposition=top}
    \renewcommand{\arraystretch}{1.5}
    \caption{Performance comparison with real-world process models}
    \label{tab:mem_time}
    \resizebox{0.8\textwidth}{!}{%
        \begin{tabular}{%
                l 
                S[table-format=3.0] 
                S[table-format=3.0] 
                S[table-format=3.0] @{\hspace{1ex}} 
                S[table-format=2.2] 
                S[table-format=2.2]
            }
            \toprule
            \textbf{Event log} & 
            \textrm{Trans.} & 
            \textrm{Places} & 
            \textrm{Nodes} & 
            \textbf{Mem.usage [MB]} & 
            \textbf{Exec.time [ms]} \\
            \cmidrule(r){1-1}\cmidrule(r){2-4}\cmidrule(r){5-6}
            \href{https://doi.org/10.4121/UUID:3926DB30-F712-4394-AEBC-75976070E91F}{BPIC~12}   & 78 & 54 & 174 & 19.97 & 5.11 \\
            \href{https://doi.org/10.4121/UUID:C2C3B154-AB26-4B31-A0E8-8F2350DDAC11}{BPIC~13\textsubscript{cp}}  & 19 & 54 & 44 & 19.76 & 1.70 \\
            \href{https://doi.org/10.4121/UUID:500573E6-ACCC-4B0C-9576-AA5468B10CEE}{BPIC~13\textsubscript{inc}} & 23 & 17 & 50 & 19.89 & 2.03 \\
            \href{https://doi.org/10.4121/UUID:3CFA2260-F5C5-44BE-AFE1-B70D35288D6D}{BPIC~14\textsubscript{f}}  & 46 & 35 & 102 & 19.90 & 3.31 \\
            \href{https://doi.org/10.4121/UUID:A0ADDFDA-2044-4541-A450-FDCC9FE16D17}{BPIC~15\textsubscript{1f}}  & 135 & 89 & 286 & 20.44 & 8.39 \\
            \href{https://doi.org/10.4121/UUID:63A8435A-077D-4ECE-97CD-2C76D394D99C}{BPIC~15\textsubscript{2f}}   & 200 & 123 & 422 & 20.91 & 12.30 \\
            \href{https://doi.org/10.4121/uuid:ed445cdd-27d5-4d77-a1f7-59fe7360cfbe}{BPIC~15\textsubscript{3f}}  & 178 & 122 & 396 & 20.77 & 11.49 \\
            \href{https://doi.org/10.4121/uuid:679b11cf-47cd-459e-a6de-9ca614e25985}{BPIC~15\textsubscript{4f}}  & 168 & 115 & 368 & 20.55 & 11.38 \\
            \href{https://doi.org/10.4121/uuid:b32c6fe5-f212-4286-9774-58dd53511cf8}{BPIC~15\textsubscript{5f}}  & 150 & 99 & 320 & 20.43 & 9.16 \\
            \href{https://doi.org/10.4121/UUID:5F3067DF-F10B-45DA-B98B-86AE4C7A310B}{BPIC~17} & 87 & 55 & 184 & 19.91 & 5.67 \\
            \href{https://doi.org/10.4121/UUID:270FD440-1057-4FB9-89A9-B699B47990F5}{RTFMP} & 34 & 29 & 82 & 19.81 & 3.47 \\
            \href{https://doi.org/10.4121/UUID:915D2BFB-7E84-49AD-A286-DC35F063A460}{Sepsis} & 50 & 39 & 116 & 19.75 & 3.65 \\
            \bottomrule
        \end{tabular}
    }
\end{table}
To this end, we employ the Inductive Miner algorithm version proposed in in~\cite{DBLP:conf/bpm/LeemansFA13}, which filters out infrequent behavior while still discovering well-structured, sound models~\cite{DBLP:journals/tkde/AugustoCDRMMMS19}. Thus, we first run the Inductive Miner on the event logs considered in~\cite{DBLP:journals/tkde/AugustoCDRMMMS19} to generate the Workflow nets. We then apply \cref{alg:imperative-to-declarative} to derive the corresponding {\Declare} specification.
%To this end, we used the Inductive Miner algorithm version proposed in \cite{DBLP:conf/bpm/LeemansFA13}, which drops infrequent behavior from logs, still discovering well-structured and sound models, as experimentally validated in~\cite{DBLP:journals/tkde/AugustoCDRMMMS19}. Thus, we first produced the Workflow nets by running the Inductive Miner on the real-world event logs considered in~\cite{DBLP:journals/tkde/AugustoCDRMMMS19}, we constructed the respective Workflow nets and subsequently applied \cref{alg:imperative-to-declarative} to generate the corresponding {\Declare} specification.
We report the aggregate test result in \cref{tab:mem_time}, detailing the memory usage, the execution time, and all the features of the mined Workflow nets. We find that the overall differences in resource usage are negligible. These real-world test outcomes again follow the complexity assumptions outlined in \cref{sec:algo:subsec:spacetime}. 

\subsection{A downstream task: Using constraints as determinants for process diagnosis}
\label{sec:conformance}
% Here we showcase a possible use of our implemented approach as a means to analyze the fit of a process model with a set of runs. %n early application in process diagnostics. Indeed, while imperative models are 
%well-established and 
% widely used within organizations, their employment in behavior diagnostics requires a major interpretative workload. %can be particularly challenging. %due to the interplay of process steps and their dependencies. 
\Cref{alg:imperative-to-declarative} enables the transition from an overarching imperative model to a constraint-based specification, enclosing parts of behavior into separate constraints. Herewith, we aim to demonstrate how we can single out the violated rules constituting the process model behavior, thereby spotlighting points of non-compliance with processes. In other words, we aim to use constraints as determinants for a process diagnosis.
%, rather than registering 
% encodes a safe and sound Workflow net into a set of behavioral constraints. % that potentially streamlines behavioral analysis.
%
\begin{figure}[tb]
    \centering
    \includegraphics[width=0.6\textwidth]{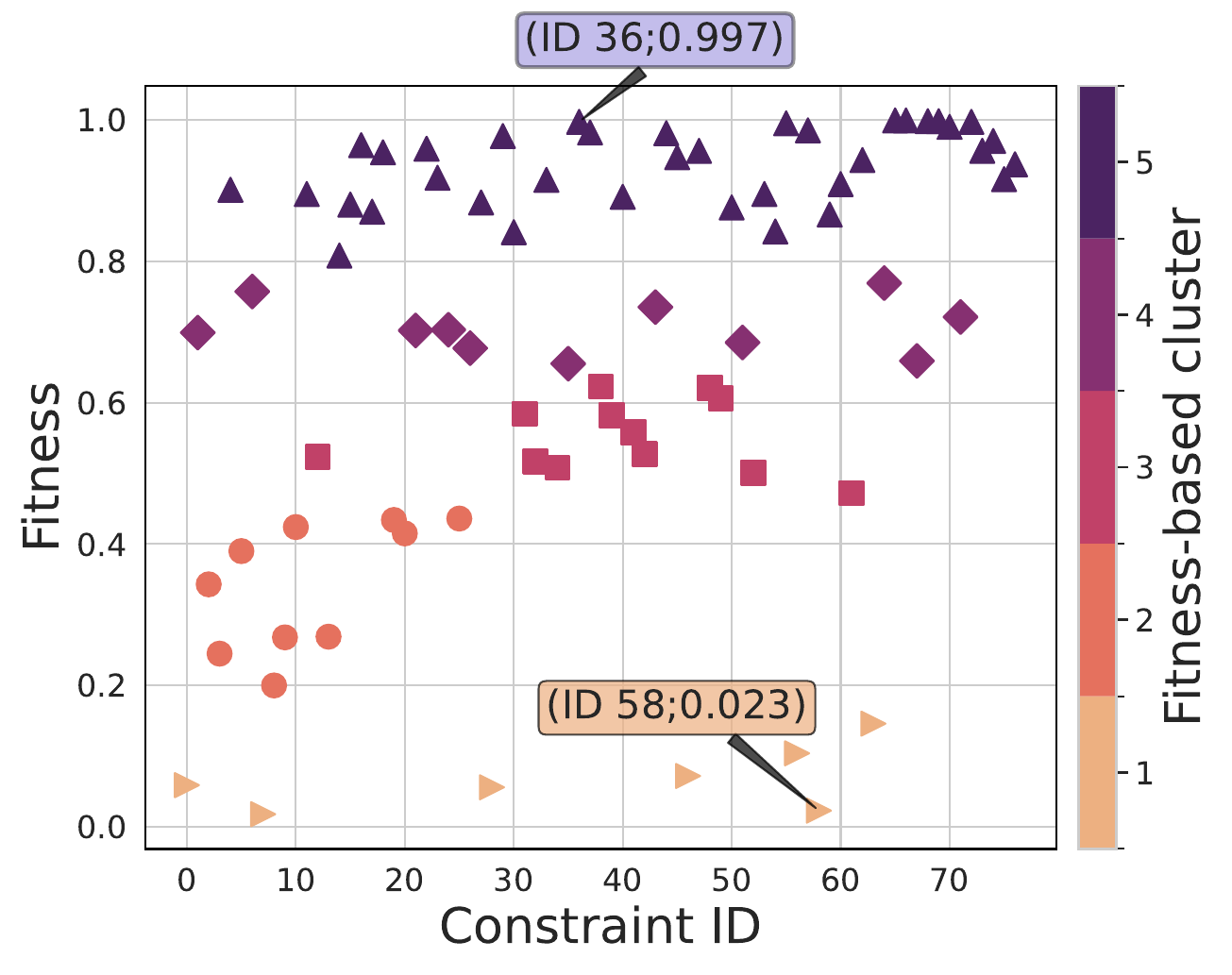}
    \caption[Scatter plot for constraints clustered by fitness values]{Fitness-based clusters of the constraints in the descriptive model of \href{https://doi.org/10.4121/uuid:b32c6fe5-f212-4286-9774-58dd53511cf8}{BPIC~15\textsubscript{5f}}}
    \label{fig:scatter}
\end{figure}
%
% In detail, we compute \emph{fitness}~\cite{DBLP:conf/bpm/LeoniMA12} for measuring the degree of correspondence of the event log with the synthesized {\Declare} specification. 
For this purpose, we created a dedicated module extending a declarative specification miner for constraint checking via the replay of runs on semi-symbolic automata like those in \cref{fig:constraint:automata:atmoone:inst,fig:constraint:automata:altprec:inst,fig:constraint:automata:end:inst}, following~\cite{DiCiccio.etal/IS2018:RelevanceofBusinessConstrainttoEventLog}. 
Without loss of generality, we build the runs from data pertaining to building permit applications in Dutch municipalities from BPIC~15\textsubscript{5f}~\cite{bpic155} %.
%BPI Challenge 2015,\footnote{\url{https://doi.org/10.4121/uuid:31a308ef-c844-48da-948c-305d167a0ec1}} 
% To build the reference model, we  
and apply the preprocessing technique mentioned in~\cite{DBLP:journals/tkde/AugustoCDRMMMS19}, resulting in \num{975} traces. We then process the log with the $\alpha$-algorithm~\cite{DBLP:journals/tkde/AalstWM04} and provide the returned net as input to our implementation of \cref{alg:imperative-to-declarative}. 

% Comparing the synthesized {\Declare} specification with the log, we determine a total average fitness of %approximately \(90\%\) (
% \num{0.905}, which suggests a good yet not full behavioral correspondence.
% All the more reason so, we are interested in understanding more in depth which rules got violated and when.
We observe that the specification consists of \num{129} constraints. Of those, our tool detected violations by at least a trace for \num{77} of those. \Cref{fig:scatter} illustrates the percentage of satisfying traces (henceforth, \emph{fitness} for brevity) of the \num{77} constraints, which we clustered into five distinct groups to ease inspection. %To examine the practical application of the synthesized {\Declare} specification in identifying anomalies in the process behavior, 
We first focus on violated constraints exhibiting high fitness (the blue upward triangles at the top of \cref{fig:scatter}). %Our analysis suggests that these alterations in constraint fitness often correspond to specific discrepancies in expected process behavior. 
Let us take, e.g., the constraint identified by ID 36 in the figure: \AltPrecShort{\{\Tx{\Task{01\_HOOFD\_490\_1}}, \Tx{\Task{13\_CRD\_010}}\}}{\Tx{\Task{01\_HOOFD\_490\_1a}}}, which exhibits a fitness of \num{0.997}. This constraint imposes that when ``\textit{Set Decision Status}'' ({\Tx{\Task{01\_HOOFD\_490\_1a}}}) occurs, it should be preceded by either ``\textit{Create Environmental Permit Decision}'' (\Tx{\Task{01\_HOOFD\_490\_1}}) or ``\textit{Coordination of Application}'' (\Tx{\Task{13\_CRD\_010}}).
In the three traces violating the constraint (11369696, 9613229, 12135936), though, ``\textit{Set Decision Status}'' is preceded by neither of the two. By further inspection, we observe ``\textit{No Permit Needed or Only Notification Needed}'' ({\Tx{\Task{14\_VRIJ\_010}}}) in the trace prefix instead, suggesting that the process bypasses the standard decision-making steps defined by the reference model in favor of an alternative where a permit decision is unnecessary. %Thus, the synthesized specification enables granular and explanatory diagnostics of this type of behavioral alteration.
%
%LOW FITNESS
%
On the other side of the spectrum, let us look at constraints with low trace fitness values (depicted by rightward orange triangles in \cref{fig:scatter}). 
%rarely satisfied in the event log, 
These constraints likely suffer from systematic defects rather than spurious alterations in the process behavior. Constraint ID 58, e.g., belongs to this group:  \AltPrecShort{\Tx{\Task{1\_HOOFD\_510\_2}}}{\{\Tx{\Task{01\_HOOFD\_510\_3}},\Tx{\Task{01\_HOOFD\_520}},\Tx{\Task{END}}\}} (depicted in the lower section of \cref{fig:scatter}). %Its fitness level is \num{0.023}. With a closer look, we observe that the violating traces do not contain occurrences of the target activities, before \Tx{\Task{END}}\} occurs. 
Other constraints in the same group
%(ranging from \num{0.018} to \num{0.146}) 
have in common the presence of \Tx{\Task{END}} in the activator's set. 
%, introduced according to \cite{DBLP:journals/tkde/AugustoCDRMMMS19} to generate the reference model. 
An explanation is that the BPIC~15 log allows a multitude of possible conclusions. %, as the underlying process is not uniquely defined. Furthermore, 
The $\alpha$-algorithm, though, disregards the occurrence frequency of individual transitions during model construction, resulting in a non-selective inclusion of all events. Consequently, this affects the fitness of those constraints. Our tool specifically pinpoints and isolates the effect of this tendency from the remainder of the net. %, which bears a relatively high support overall.
%Allowing to understand which behavioral trajectory is the most respected
% We plan to carry out further investigations to t
Thoroughly assessing the suitability %of extending #
of our approach for process diagnostics transcends the scope of this paper but paves the path for future work.

\section{Related work}
\label{sec:related}
The relationship between imperative and declarative modeling approaches has been extensively explored in the existing literature, with a prevailing focus directed toward the development of analytics tools that effectively compare and integrate the strengths of both paradigms. 
%This paper is mostly inspired by the aforementioned work of Prescher et al.~\cite{DBLP:conf/simpda/PrescherCM14} 
%wherein the authors introduce a systematical transformation approach of the declarative specification into behaviorally equivalent safe Petri net. 
%Drawing inspiration from the contributions that have been made towards the purpose of establishing a formal connection between the declarative and imperative paradigms
Building on previous contributions aimed at establishing a formal connection between these paradigms~\cite{DBLP:conf/simpda/PrescherCM14,DBLP:conf/apn/CosmaHS23}, our research focuses on providing a systematic approach for translating safe and sound Workflow nets into their declarative counterparts. 
A growing research stream configures this transformation aiming to leverage the support provided by the declarative specifications for conformance checking and anomaly detection. Notably, integrating declarative constraints into event log analysis facilitates more comprehensive diagnostics than those provided by trace replaying techniques. %~\cite{DBLP:conf/edoc/AdriansyahDA11}.
%, whose interpretation is inherently context-dependent and lacks determinism. 
%
% CONFORMANCE
%
In this regard, Rocha~et~al.~\cite{DBLP:conf/bpm/RochaZA24} propose an automated method for generating conformance diagnostics using declarative constraints derived from an input imperative model. Their method relies on a library of templates internally maintained in the tool. Eligible constraints are generated by verifying the instantiation of those templates against the model's state space. The ones that are behaviorally compatible are then subject to redundancy removal pruning. %verified across all traces. Redundancy pruning techniques are then applied, yielding a refined set of constraints employed to generate readily interpretable diagnostics, thus enabling a more flexible analysis by capturing high-level behavioral patterns.
Finally, the retained constraints are checked for conformance against log traces.
In \cite{DBLP:conf/bpm/BergmannRK23}, the authors present a tool that derives a set of eligible constraints directly extracting relations based on a selection of BPMN models' activity patterns. 
%While the application is of particular interest in practice, there is a lack of formal support for the conversion, resulting in a non-unique parallelism. 
The work of Rebmann~et~al.~\cite{DBLP:journals/corr/abs-2407-02336} proposes a framework for extracting best-practice declarative constraints from a collection of imperative models aiming to discover potential violations and undesired behavior. Constraints are extracted %by instantiating each declarative template with all possible parameters akin 
akin to~\cite{DBLP:conf/bpm/BergmannRK23}, then refined and validated via natural-language-processing techniques to measure their relevance for a given event log. Busch~et~al.~\cite{DBLP:conf/bpm/BuschKL24} % emphasizes the explainability of deviations by proposing a sequence-to-sequence approach to specification generation.
adopt a similar technique to check constraints characterizing process model repositories against event logs.
%
% The equivalence established by these contributions does not constitute a formal connection between the two paradigms. 
All these techniques share our aim to derive declarative constraints from imperative models given as input. However, they do so via simulation or state space exploration, with limited guarantees of behavioral equivalence. % approaches and subsequently trimmed through different pruning techniques. 
In contrast, our work proposes an algorithm that is proven to establish a formal equivalence between the given imperative model and the derived declarative specification. Being based on the sole exploration of the net's structure, it is also lightweight in terms of computational demands. 

\section{Conclusion and future work}
\label{sec:conclusion}
In this paper, we presented a systematic approach to translate safe and sound Workflow nets into bisimilar {\Declare} specifications. The latter are based solely on three {\LTLf} formula templates from the {\Declare} repertoire with branching: \AtMoTmp, \EndTmp, and \AltPrecTmp. We provide a proof-of-concept implementation, of which we evaluate scalability and showcase applications against synthetic and real-world testbeds.

We believe that the scope of this research may be expanded in a number of directions. A natural extension of our work is the inclusion of label-mappings of the Workflow net in the declarative specifications, which would turn the constraints' semi-symbolic automata into transducers that are advantageous in conformance checking contexts.%\todo{CDC4@Marco: check if this is conveying the right message to the right reviewers' class.}
Moreover, we seek to broaden the application of our solution to detect behavioral violations, extending support to a wider range of imperative input models. Also, we aim to investigate the correspondence between specification inconsistencies and Workflow net unsafeness and/or unsoundness. Another promising application lies in hybrid representations combining imperative and declarative paradigms. In this regard, our approach could facilitate behavioral comparisons akin to~\cite{DBLP:conf/sose/Baumann18} and enable the construction of hybrid representations tailored to diverse scenarios~\cite{DBLP:journals/bise/SmedtWVP16}. 
%
%
%Currently, utilizing a Workflow net that contains silent transitions, we generate constraints that, although algorithmic correctness, would not entirely match the trace in an event log. Using silent transition reduction techniques and optimizing constraint parsing for specific instances may serve to overcome this limitation. 
%
 %, evaluate its e performance and compare it performance with state-of-the-art benchmarks. 
%We also plan to integrate additional conformance metrics into the current fitness we have been using. %~\cite{DBLP:journals/is/CecconiBCS24}
%\todo[inline]{Deadlock/livelock con AltResp\\}

\begin{comment}
\subsubsection{Acknowledgments.}
\begin{sloppypar}
	This work has been partly funded by MUR under PRIN grant B87G22000450001 (PINPOINT), the Latium Region under PO~FSE+ grant B83C22004050009 (PPMPP), and \textbf{Health-e-Data}.
	\todo{With a double-blind submission, Ack's are not allowed.}
\end{sloppypar}
\end{comment}

\bibliographystyle{splncs04}
\bibliography{bibliography}

\end{document}